\documentclass{article}

\usepackage{makecell}
\usepackage{amsfonts,amsthm,amssymb}
\usepackage{bm}

\usepackage{amsmath}
\usepackage{booktabs} % For pretty tables
\usepackage{tablefootnote} % for table footnotes
\usepackage{caption} % For caption spacing
\usepackage{subcaption} % For sub-figures
\usepackage{graphicx}
\usepackage{pgfplots}
\usepackage[all]{nowidow}
\usepackage[utf8]{inputenc}
\usepackage{tikz}
\usepackage{multicol}
\usepackage{algpseudocode,algorithm,algorithmicx}

\usepackage{multirow}
\usepackage[inline]{enumitem} % Horizontal lists
\usepackage[hyphens]{url}
\usepackage[breaklinks=true]{hyperref}

\usepackage[margin=1in]{geometry}

\newtheorem{assumption}{Assumption}
\newtheorem{lemma}{Lemma}

\definecolor{blue}{HTML}{1F77B4}
\definecolor{orange}{HTML}{FF7F0E}
\definecolor{green}{HTML}{2CA02C}

\pgfplotsset{compat=1.14}

\title{Analyzing Disparity and Temporal Progression of Internet Quality through Crowdsourced Measurements with Bias-Correction}
\author{Hyeongseong Lee, %\and %\inst{1} 
Udit Paul, % \and %
Arpit Gupta, % \and %\inst{3}
Elizabeth Belding, % \and %\inst{4}
 and Mengyang Gu\thanks{Corresponding author (\href{mailto:mengyang@pstat.ucsb.edu}{mengyang@pstat.ucsb.edu})
 }\\ %\footnote{} %\inst{5}
 ~~~
 \\
  University of California, Santa Barbara, California, USA
}
\date{}
\begin{document}
\maketitle              % typeset the header of the contribution

\begin{abstract}
%{\bf abstract}:
 Crowdsourced speed test measurements are an important tool for studying internet performance from the end user perspective.  Nevertheless, despite the accuracy of individual measurements, simplistic aggregation of these data points is problematic due to their intrinsic sampling bias. In this work, we utilize a dataset of nearly 1 million individual Ookla Speedtest\textsuperscript{\tiny\textregistered} measurements, correlate each datapoint with 2019 Census demographic data, and develop new methods to present a novel analysis to quantify regional sampling bias and the relationship of internet performance to demographic profile. We find that the crowdsourced Ookla Speedtest data points contain significant sampling bias across different census block groups based on a statistical test of homogeneity. We introduce two methods to correct the regional bias by the population of each census block group. Whereas the sampling bias leads to a small discrepancy in the overall cumulative distribution function of internet speed in a city between estimation from original samples and bias-corrected estimation, the discrepancy is much smaller compared to the size of the sampling heterogeneity across regions.  Further, we show that the sampling bias is strongly associated with a few demographic variables, such as income, education level, age, and ethnic distribution. Through regression analysis, we find that regions with higher income, younger population, and lower representation of Hispanic residents tend to measure faster internet speeds along with substantial collinearity amongst socioeconomic attributes and ethnic composition. Finally, we find that average internet speed increases over time based on both linear and nonlinear analysis from state space models, though the regional sampling bias may result in a small overestimation of the temporal increase of internet speed. 

%\keywords{bias correction, crowdsourced measurements, data integration, digital equity, speedtest, temporal analysis, internet performance, internet measurement}

\end{abstract}

\section{Introduction}

The allocation of large US federal grants and subsidies to expand Internet infrastructure, such as through the \$42.5 billion Broadband Equity Access and Deployment (BEAD) program initiated in November 2021, is an important step in addressing digital inequity and providing high-quality Internet access to all Americans.  To properly allocate this funding to regions of greatest need, there must first be a methodology for measuring the current state of Internet access and quality at a given location.  Indeed, the US Federal Communications Commission (FCC) recently undertook a significant effort to update publicly available information about US broadband infrastructure.  
First, the FCC worked to improve accuracy in the national Broadband Serviceable Location Fabric, defined as ``a common data set of all residential and business locations (or structures) in the U.S. where fixed broadband internet access service is or can be installed."  Based on this new Fabric, they updated the National Broadband Map with ISP-provided data on the maximum download speed available at each location~\cite{fcc_map}.    
While this is a significant advancement in publicly accessible data on broadband infrastructure, the Fabric and the map  have a number of drawbacks.  
In particular, the Broadband Map does not report actual or average performance but theoretical maximums. It also does not provide information about the reliability or stability of the access over time. Finally, providers have been repeatedly found to overstate coverage claims~\cite{broadband_speed_fcc,major2020}.

To address these drawbacks, the research community, individual users, and governments and policymakers have turned to crowdsourced ``speed test" active measurement data, such as Ookla's speedtest.net~\cite{ookla_main}, Measurement Lab's speed.measurementlab.net~\cite{mlab_speed}, FAST~\cite{fastapi} and Xfinity's speed test~\cite{xfinity_speed}. 
These platforms are in wide use globally.  For instance,  Ookla reports a daily average of over 10 million Speedtests~\cite{aboutookla}.
Speed test platforms offer a variety of measurements, such as instantaneous upload/download speed and latency and loaded latency data, through either an app or a website.  %Aggregated together, these measurements offer  a promising starting point for understanding actual Internet connectivity.
%The current state-of-the-art draws heavily on crowdsourced active network measurements, e.g. ``speed tests".  
As such, these platforms provide an important snapshot of the network state from the vantage point of the end users. Further, because they are active measurements, they provide data on actual performance instead of the theoretical maximum performance reported by the providers. %Popular network speed test platforms, such as Ookla's speedtest.net~\cite{ookla_main}, Measurement Lab's speed.measurementlab.net~\cite{mlab_speed}, FAST~\cite{fastapi} and Xfinity's speed test~\cite{xfinity_speed}, are utilized by Internet users worldwide to conduct these measurements.  
Because of the inherent benefits, numerous U.S. government initiatives (e.g.~\cite{brindisi_report,calspeed,grownorthWI,alabamaSpeed,MImoonshoot,ruralPA}) utilize crowdsourced speed test data to help map the internet access landscape.  %With this data, local governments, community organizations, and others can attempt to discern where to make the economic investment in infrastructure to address digital inequality. {\em Perhaps most critically, the FCC itself has recently specified a  challenge process~\cite{fcc_cognizable}, whereby individual users and communities can gather active measurement data to challenge provider-reported coverage claims.}   

%Internet performance, often referred to as internet speed, has been considered paramount encompassing many realms including business, education, health, and occupational environment. In addition to accessibility to internet, the distribution of internt speed and its inequality across socio-demographic groups have received substantial attention since the COVID-19 pandemic. To measure the quality of internet performance, crowdsource-based speed test program such as Ookla and Measurement Lab (M-Lab) has been heavily used for recent five years. A series of study has been conducted about crowdsourced internet speed tests. For instance,  

%Internet speed tests have been widely used globally. Ookla, for instance, reports a daily average of over 10 million speedtests  \cite{aboutookla}. 
However, despite the broad use speed tests, there are multiple critical limitations that can 
affect the estimation accuracy of these crowdsourced measurements. 
%and usability of the aggregated data. %it is not clearly known to the policy makers on how to perform statistical inference using this data, due to a few inherent limitations~\cite{}.  
In particular, due to the nature of crowdsourced data, the sampling mechanism is often uncontrollable~\cite{saxon2022chicago}.  As a result, inherent in this data is a variety of biases along multiple dimensions that can skew the measurement results in unanticipated ways. For instance, the sampling biases can lead to inappropriate conclusions if they are treated as simple random samples~\cite{meng2018statistical}. Another barrier is that consumers' information, such as broadband plans and demographic profiles, is not directly available from the speed tests. Even if the broadband plans can be correctly inferred for the speed test with the use of other data sets~\cite{paul2023context}, the lack of a demographic profile associated with each measurement makes it challenging to correct the sampling bias directly. Finally, speed test performance may depend on the choice of test server; prior work has found that some servers consistently report speeds 10\% lower than other servers~\cite{macmillan2023comparative}. Nevertheless, the vast number and diversity of speed test data points present abundant information for understanding the disparity of internet quality, particularly as related to different socio-demographic groups and the change/evolution of internet quality over time. 
%, which, to the best of our knowledge, were not studied. 
%Though these caveats were gradually realized by the community, to the best of our knowledge, there is no study to quantify whether the sampling bias leads to 
%there have not been in-depth discussions on inequity of internet performance and its measurement in light with correction of sampling bias. 
%\emb{it would be great if we could have a For instance here with an example.}

It is within this context that we perform our analysis. Specifically, we utilize a vast corpus of nearly 1~million individual Ookla Speedtest\textsuperscript{\tiny\textregistered} measurements, to which we have access through an Ookla for Good\textsuperscript{TM} Data Use Agreement, to correlate network performance measurements and trends with demographic profiles of census block groups from the 2019 U.S. census data.  Our goals are to characterize disparities in internet quality based on socio-demographic groups and to analyze the change in internet quality over time. We illustrate our novel methods using four representative cities and two device types (iOS and Android) as examples, corresponding to 978,101 Speedtest data points over 580 days.
%, corresponding to 80,822 number of tests over 580 days.
Specifically, the contributions of our work are threefold: 
\begin{enumerate}
    \item {We develop novel methods for correcting regional sampling bias and correlating this sampling bias with demographic profiles in the population.} By utilizing the chi-squared test, we find that the proportion of the Speedtests in different census block groups significantly deviates from the baseline proportions in the population. 
    %The sampling is significantly deviant from reference population. 
    We introduce re-weighing and re-sampling methods for correcting the regional sampling bias. Even though we found a visible discrepancy between the original and corrected estimation of the internet distribution, this difference is much smaller than the regional bias, meaning that the sampling bias itself may not substantially affect the estimation of the internet speed\footnote{In this paper, we use the terms ``internet speed" and ``measured internet speed" interchangeably.  In either case, we are specifically referring to the download speed of the connection to the user end device as measured through an Ookla Speedtest.}  distribution in a city. 
    %However, this bias does not severely distort the estimation of the  internet speed distribution. 
    We study the underlying reasons for disproportionate sampling and find that the sampling bias is strongly associated with a few demographic variables, including
income, education level, age, and ethnic distribution.
%is related to regional asymmetry of demographic features. Regions with higher median income, overall education level, more white residents and less black residents are likely to submit more speed tests.

    %To do so, we perform exploratory data analysis and run statistical tests, such as the chi-squared test, to compare the distributions of the speed tests and the proportion of the population for hundreds of census blocks of multiple cities, which contain millions of internet speed test measurements. The results of the test show extremely strong heterogeneity of upload and download speed of the Ookla sample across different census blocks in each city, which means that a large subsampling bias is contained in the data. We further developed a subsampling approach to correct the regional bias within the sample and estimate the cumulative distribution function of the original sample and bias-corrected sample for each city. The result shows the noticeable differences between the bias-corrected and original samples, but the differences are not large.

    \item {We introduce variable selection and regression techniques to fuse data at different spatial granular levels for understanding how internet quality varies across demographic profiles.
    %Analysis and reappraisal of internet speed correlated with demographic attributes with regional bias correction.
    } By studying the individual and aggregated data, we show that the use of individual-level measurements is more statistically efficient in estimating the variability of the speed test measurements than the aggregated data. Through backward variable selection for the original and bias-corrected samples, we find that regions with higher income, younger populations, and smaller Hispanic populations tend to have higher internet speeds. Furthermore, we find strong collinearity between a few demographic variables, which tend to affect the outcomes of the fitted model jointly. 
    %regions with higher income and education level, younger population and smaller Hispanic population tend to have higher internet speed. Furthermore, we find strong collinearity between a few demographic variables, which  tend to jointly affect the outcomes of the fitted model. 
    %Furthermore, Regional sampling bias tends to underestimate the positive correlation of education level with internet speed.

%Our second goal is to correlate the demographic variables at the census block level, such as income and race, to the quality of the internet using both the original sample and bias-corrected sample. 
%We show using individual-level measurements is much more efficient in estimating the variability of the speed test than the aggregated data. Furthermore, utilizing the backward variable selection, we found a variety of variables that are significant to the observations. 

    \item {We analyze temporal changes of measured internet performance over time.} We conduct both a linear regression model and a state space model to study the temporal change of the internet measurements. We utilize the state space representation of the Gaussian process with a Mat{\'e}rn function having a half-integer roughness parameter, which reduces the cost of computing the likelihood and making predictions from $\mathcal O(n^3)$ to $\mathcal O(n)$ operations without approximation, with $n$ being the number of observations. The acceleration with no approximation enables the approach to be computationally feasible with a massive number of crowdsourced measurements. Both linear and nonlinear analyses suggest that internet speed quality gradually improves over time. 
    %regional sampling bias may tends to slightly overestimate the overall gradients.
\end{enumerate}

The rest of this paper is organized as follows. In Section \ref{sec:data}, we introduce the two primary data sets used in this study. In Section \ref{sec:bias-detection},  
%Section~\ref{subsec:sample_balance} 
we investigate the regional sampling bias and the association with demographic variables. We introduce two ways for correcting the sampling bias and compare the difference of estimation between the original and bias-corrected samples. 
Section~\ref{sec:lm} introduces regression analysis of measured internet speed using demographic profiles at different census block groups for both the original and bias-corrected samples. 
Section~\ref{sec:temporal} studies the change of measured internet speeds over time using both linear and nonlinear estimation. We conclude the paper in Section ~\ref{sec:conclusion}. We provide proofs and report additional results in the Appendix.  

\section{Description of data}\label{sec:data}

We combine Ookla\textsuperscript{\tiny\textregistered} internet speed measurement data~\cite{ookla_main} and demography data provided by the U.S. Census Bureau for our study. In the following sections, we describe each dataset in detail.
%For our study, we use Ookla\textsuperscript{\tiny\textregistered} internet speed measurement data~\cite{ookla_main} combined with US census database based on census block groups. We describe each dataset in detail %and explain connections between two datasets 
%in the following sections.
 
\subsection{Ookla's Speedtest}\label{subsec:ookla}

Ookla's Speedtest\footnote{http://speedtest.net} (data provided through Ookla’s Speedtest Intelligence\textsuperscript{\tiny\textregistered})  possesses over 16k measurement servers worldwide~\cite{ookla_server} and allows users to assess the quality of their Internet connection using either a web-based portal or native mobile application~\cite{ookla_main}.  
For each Speedtest, a nearby test server is selected, and  (potentially multiple) TCP connections are used to calculate the throughput of the path. Ookla’s Speedtest Intelligence dataset contains individual Speedtest measurements that include QoS metrics (up/down throughput, latency, packet loss, jitter), as well as meta-features such as ISP, device type, and access type.  Ookla provides performance data aggregated over time and space to the public~\cite{ookla-open-github}.  
However, our Data Usage Agreement (DUA) with Ookla provides us access to nearly 1 million individual Speedtest measurements from four major metropolitan cities in the U.S., which we use for this study. Each of these cities has a population in the range of 370---650k. 

We use the data obtained through our DUA to investigate the download speed data from two different device types---iOS and Android.  We focus on these two device types due to the richness of metadata and geographical information associated with each data point, as well as due to the widespread popularity and usage of these platforms. For demonstration purposes, we use two representative cities for our detailed analysis. Due to our DUA, we maintain the confidentiality of these cities and only refer to them as City A and City B. We also provide numerical analysis of two additional cities, referred to as City C and City D,   in Appendix~\ref{appendix-sec:othercities}. %The Speedtest measurements are collected from 05-31-2020 to 12-31-2021 for City A and B, and from 01-01-2021 to 12-31-2021 for City C and D. 
Since precise geographic location is crucial for our study, we ensure each measurement includes latitude/longitude data recorded at the time the Speedtest is conducted;  truncated GPS coordinates are reported for users
who allow location sharing on their equipment. We utilize the geographic location of each Speedtest to associate the test with the relevant census block group.
Table~\ref{tab:ookla} shows summary statistics for the data used in this paper. We observe that in every listed city, iOS devices outnumber Android devices in our dataset, yielding almost double the number of Speedtest data points from iOS compared to Android. 

\begin{table*}
\centering
  \caption{Overview of Ookla data.}
  \label{tab:ookla}
  \begin{tabular}{ccccc}
    \toprule
    City & Device type & \# of unique devices & \# of tests & Date range\\
    \midrule
    \multirow{2}{*}{A} & iOS & 23,649 & 165,229 & \multirow{2}{*}{May 31, 2020 -- Dec 31, 2021}\\
                       & Android & 8,438 & 73,094 & \\
    \midrule
    \multirow{2}{*}{B}& iOS & 15,507 & 92,526 & \multirow{2}{*}{May 31, 2020 -- Dec 31, 2021}\\
                        & Android & 8,553 & 70,039 & \\
    \midrule
    \multirow{2}{*}{C} & iOS & 34,005 & 213,770 & \multirow{2}{*}{Jan 1, 2021 -- Dec 31, 2021}\\
                       & Android & 14,730 & 117,806 &\\
    \midrule
    \multirow{2}{*}{D}& iOS & 20,607 & 146,703 & \multirow{2}{*}{Jan 1, 2021 -- Dec 31, 2021}\\
                        & Android & 11,122 & 98,934 &\\                    
    \bottomrule
  \end{tabular}
\end{table*}

\subsection{Demographic data from the U.S. Census Bureau}\label{subsec:census}

%We collect demographic data from the 2019 U.S. Census Bureau. The U.S Census Bureau provides extensive demographic profiles across census blocks in the U.S. 
%including but not limited to: average income, bachelor degree, population size, age, sex, and race. 
%Hispanic or Latino origin, household type, relationship to householder, group quarters population, housing occupancy, and housing tenure. 
To study demographic trends with performance data, we leverage data from the U.S. Census Bureau.  In particular, 
we select the 2019 Census demographic information due to the inaccuracy of data since the COVID-19 pandemic~\cite{census2022}.
For City A and B, we obtain population data for each census block group as well as the following demographic attributes: income, age, sex, education level, internet penetration rate, and ethnicity distribution. Income and age are represented by the median household income and median age, respectively. Sex information is characterized by the proportion of males within the total population of each census block group. The level of education is quantified as the proportion of individuals with a bachelor's degree or higher among the population aged 25 and older. The internet penetration rate is the percentage of households with an internet subscription out of the total number of households of each census block group. 
%with presence of internet. 
Ethnic distribution is summarized by the proportion of white, black, Asian, and Hispanic populations within the total population of each census block group.

 The demographic information from the U.S. Census Bureau is only available at the level of the census block group. %That is, we are not able to access individual demographic profiles that corresponds to certain internet speed measurement from Ookla's Speedtest. 
Therefore, we build an integrated dataset with a spatial granular structure, where speed measurements identified in each census block group are correlated with the demographic profiles for that census block group based on U.S. Census Bureau data.

\section{Regional sampling bias detection and correction}\label{sec:bias-detection}

In this section, we explore the presence of sampling bias across regions within cities  A and B. In the context of a given city, a ``region" refers to a census block group, which constitutes a sub-area within that city. If the distribution of the number of samples does not align proportionally with the population, then we say there exists regional sampling bias within our sample, which can render our statistical analysis divergent from the truth at the city level.

\subsection{Comparison of sample and population distributions at the census block group level}
\label{subsec:sample_balance}

Each Ookla Speedtest measurement is associated with a certain census block group based on the geographic location of the test. By comparing the distribution of the sample sizes across census block groups with their population distribution,  we can test whether the sample sizes are homogeneous with the population across regions. As described in Section~\ref{sec:data}, the population of each census block group is based on data from the U.S. Census Bureau. 
We plot the cumulative distributions and probability mass functions between the sample and population at the level of census block group in  Fig.~\ref{fig-pop-sam} for City A and  B, where the census block groups are ordered according to their population sizes. We observe that there is a recognizable discrepancy between the distribution of sample sizes from both Android and iOS in Ookla Speedtests and the population at each census block group. The comparison suggests that there is a regional sampling bias in the Speedtest measurements, as certain census block groups contain disproportionately larger samples than others.  

%Figure shows the distribution heat map of real sample (number of tests) and actual population across all the census block in New Orleans. It turns out that the distribution of sample sizes is different from that of population sizes, regardless of the sort of devices. We can find some census blocks with less or more number of tests than they should have based on the population proportion. 

\begin{figure}[t]
  \centering
  \includegraphics[width=.8\linewidth]{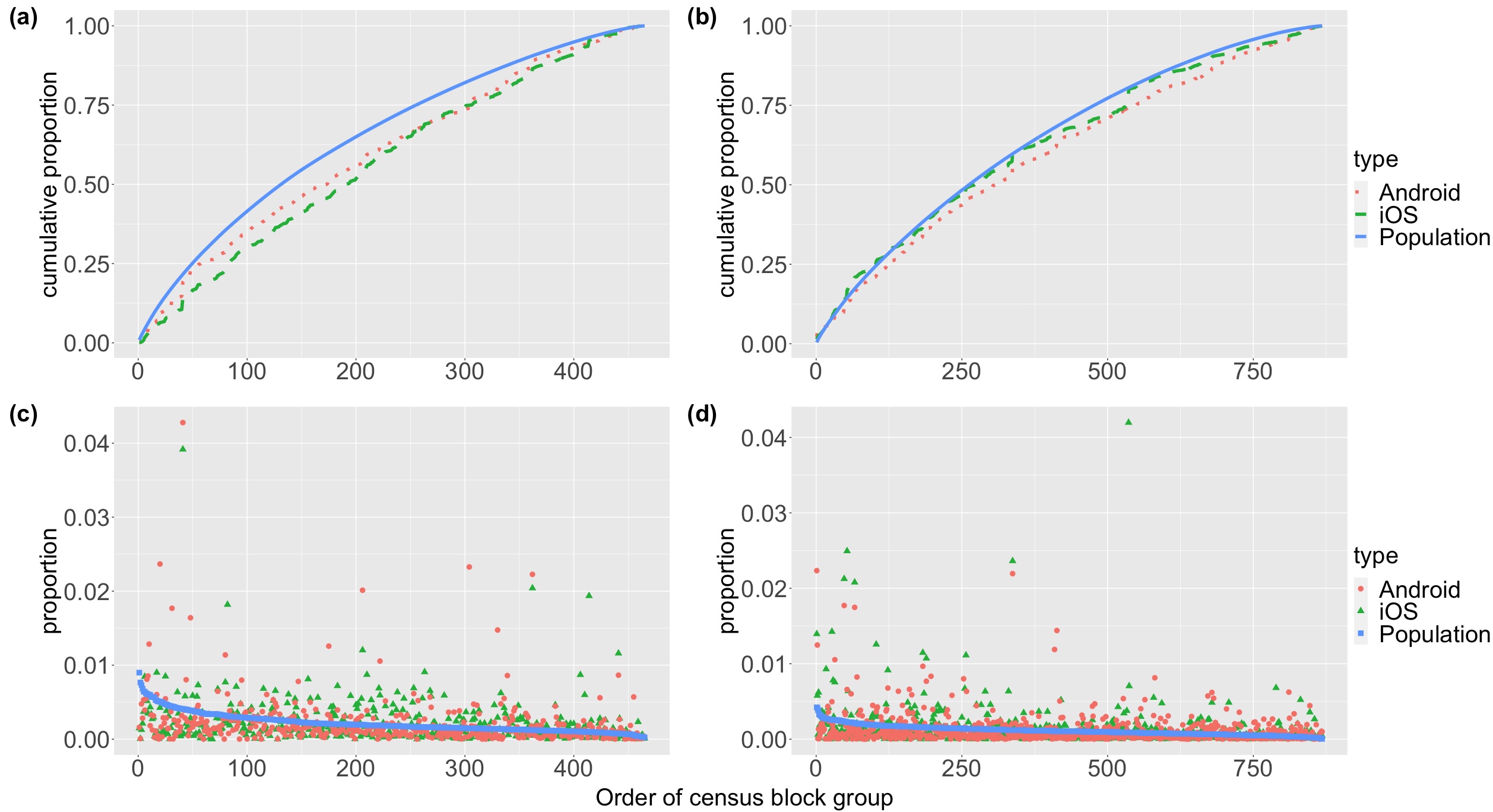}
  \caption{\label{fig-pop-sam}Comparison of the distributions of population and sample sizes. Census block groups are ordered left to right from largest to smallest population.  Cumulative distribution functions are represented for City A and B in (a) and (b), respectively, while probability mass functions for each city are 
   in (c) and (d), respectively. The x-axis indicates the rank of the population size in each census block group, and the census block group with the largest population has rank 1. } 
  %Comparison of the distributions of population and sample sizes for collected regions. Each of plots is a heat map of A-(a) population of New Orleans from iOS devices; A-(b) sample sizes of New Orleans from iOS devices; B-(a) population of New Orleahttps://www.overleaf.com/project/648b6e1c2cdf8ecc153148aans from Android devices; B-(b) sample sizes of New Orleans from Android devices; C-(a) population of Detroit from iOS devices; C-(b) sample sizes  of Detroit from iOS devices; D-(a) population of Detroit from Android devices; and D-(b) sample sizes of Detroit from Android devices.}
  %\Description{A-(a) population heat map of New Orleans from iOS devices; A-(b) sample sizes heat map of New Orleans from iOS devices; B-(a) population heat map of New Orleans from Android devices; B-(b) sample sizes heat map of New Orleans from Android devices; C-(a) population heat map of Detroit from iOS devices; C-(b) sample-size heat map of Detroit from iOS devices; D-(a) population heat map of Detroit from Android devices; D-(b) sample-size heat map of Detroit from Android devices.}
\end{figure}
\vspace*{0.1in}

%\subsubsection{Statistical test for homogeneity of distribution}
We utilize the chi-squared test, one of the most frequently used tests for homogeneity of the samples, to quantify whether the discrepancy between the population and sample size distributions across census block groups is significant. The two assumptions of the test are summarized as follows: 
%Before employing an statistical test to examine whether the discrepancy between sample and population distribution at census blocks is significant,
%is due to the sampling error or due to an actual mechanism of regional bias. 
%we first highlight the assumptions for the test. 
\begin{assumption}\label{assum-1}
The internet speed tests are independent within and between census block groups.    
\end{assumption}

\begin{assumption}\label{assum-2}
The internet speed tests are identically distributed within each census block group, and they are representative for each census block group.
%(census block groups).
\end{assumption}

Assumption \ref{assum-1} approximately holds in the data set, but it can be affected by internet traffic and availability. For instance, an internet outage event may affect more than one census block group, making speed tests not independent. Second, certain groups of internet users may be more likely to utilize Speedtest; the relationship between internet speed and demographic profiles of the census block groups will be studied in Section \ref{subsec:lm}. 

We collect the sample of Ookla Speedtests from $k$ census block groups. For the $i$-th census block group, define $n_{i}$ and $N_{i}$ as the sizes of its sample and population, respectively, for $i=1,\cdots,k$. Furthermore, let $n$ and $N$ respectively be the total number of samples and the population among all the regions, i.e., $n = \sum_{i=1}^{k}n_{i}$ and $N = \sum_{i=1}^{k}N_{i}$. When the samples are drawn according to the population distribution, i.e., the null distribution,  the expected number of samples in the $i$-th region is $\frac{nN_i}{N}$. 
Then, the $\chi^{2}$ test statistic is defined by: 
\begin{equation}
    W = \sum_{i=1}^{k}\frac{(n_i-nN_i/N)^2}{nN_i/N}.
\end{equation}
%When there is no regional bias of among the samples,
Under the null hypothesis, where there is no regional sampling bias across different census block groups, $W\sim \chi^{2}_{k-1}$, i.e., the test statistic $W$ is distributed as a chi-squared distribution with $k-1$ degrees of freedom. 
 %under Assumptions~\ref{assum-1} and~\ref{assum-2} with degree of freedom $k-1$, which is represented by $\chi^{2}_{k-1}$. 
 %Since we are comparing two distributions (sample versus population), the degree of freedom is determined only through the number of categories in the distribution, i.e., regions (census block groups). As there are $k$ different regions, the degree of freedom is given by $k-1$.
 %At 95\% significance level, if the test statistic $w$ (realized value of the test statistics) is greater than $\chi^{2}_{k-1, 0.95}$ with $\chi^{2}_{k-1, 0.95}$ being the number such that $\mathbb{P}\left(W<\chi^{2}_{k-1, \alpha}\right) = \alpha$, or if the p-value is less than 0.05, then we reject the null hypothesis, which leads us to the conclusion that the distribution of sample sizes are  statistically different from that of population sizes.

%\MG{Any result on chi-squared test?}

\vspace*{0.1in}
\begin{table*}
\centering
  \caption{The results from $\chi^{2}$ homogeneity test.}
  \label{tab:chi-test}
  \begin{tabular}{ccccc}
    \toprule
    City & \# of regions & Device type & Test statistics $W$ & p-value\\
    \midrule
    \multirow{2}{*}{A} & \multirow{2}{*}{465} &                   iOS & 342,661 & $<10^{-16}$\\
              &       & Android & 169,012 & $<10^{-16}$\\
    \midrule
    \multirow{2}{*}{B}& \multirow{2}{*}{868} &                  iOS & 398,812 &                                 $<10^{-16}$\\
                &    & Android & 167,411 &     $<10^{-16}$\\
    \bottomrule
  \end{tabular}
\end{table*}
We find that the disparity observed in Fig.~\ref{fig-pop-sam} is in line with the results from the statistical testing as shown in Table~\ref{tab:chi-test}. We observe that the p-values from both iOS and Android devices in City A and B are significantly low, suggesting substantial inhomogeneity in the sample sizes relative to the population across regions.

%In total, City A and B contain 465 and 525 regions,  respectively. In City A, the test statistics turn out to be 342,661 for iOS devices and 169,012 for Android devices with degree of freedom 464. In City B, we obtain 780,884 and 128,545 as the test statistics for iOS and Android devices, respectively. In any case, regardless of cities and device types, p-values are significantly small (less than $10^{-16}$), which implies that there are statistical evidences that the sample size distribution and population distribution are different.

%There are 460, 461, and 376 regions collected for iOS, Android, and desktop, respectively. We find that, for every device type (iOS, Android, and desktop), the sample distribution is significantly different from the population distribution for the collected regions (p-values are less than 0.001 for all three types).

\subsection{Cumulative distribution functions of internet speed: before and after regional bias correction}\label{subsec:cdf}
Because we found significant evidence of %sampling of the Ookla data being 
the Speedtest distribution being
disproportionate across regions, we  introduce 
two ways to correct the regional sampling bias. By comparing the estimation of the cumulative distribution function from the original and bias-corrected samples, we can assess whether this sampling bias affects the estimation of the distribution of internet speed.
%Given that the distribution of internet speed is heterogeneous across regions, the estimation of city-level distribution of internet speed is susceptible to the regional bias. 
%In this section, we correct regional bias by re-sampling and examine whether the estimation of internet speed distribution is different from one before correcting regional bias.

%Note that the sample sizes different from the expectation does not necessarily imply that the sample distribution of internet speed distorts the true distribution. If the distribution of internet speed is homogeneous across the collected regions, then the sample distribution of internet speed would represent its population distribution well regardless of the imbalanced sample sizes against the population. However, it is well documented that there is noticeable regional inequality of internet performance. If it is the case across the census blocks groups in City A and B, then the estimation of internet speed is to be susceptible to the regional bias. As such, we examine the potential impact of the regional bias discussed in Section~\ref{subsec:sample_balance} by comparing the empirical CDF from the sample to distribution functions with bias correction processes applied. We introduce two different methods of bias correction: re-weighing and re-sampling method.

To delineate the empirical CDF of internet speed, let $y_{ij}$ represent the internet speed measurement from the $j$-th unit of the $i$-th census block group for $j=1,\cdots, n_i$  and $i=1,\cdots, k$. 
%($j=1,\cdots, n_i,\ i=1,\cdots, k$). 
Since our estimation covers $k$ different census block groups, we define $y$ as a simple random sample of internet speed from a collection of these $k$ census block groups. Based on  Assumptions~\ref{assum-1} and~\ref{assum-2}, the simple random sample $y$ can be represented by a mixture of $k$ random variables $y_{1},y_{2},y_{3},\cdots, y_{k}$, which denote the internet speed measurement from the 1st, 2nd, 3rd, $\cdots,$ and $k$-th census block group, respectively. Without any regional sampling bias, a simple random sample $y$ of internet speed can be expressed as 
\begin{equation}\label{regional-mixture}
    y =\sum_{i=1}^{k}y_{i}I(z_i=1),\ \text{with }\mathbf{z}=(z_{1},\cdots,z_{k})^{T}\sim\text{Multinom}\left(1, (N_{1}/N, \cdots, N_{k}/N)\right),
\end{equation}
where $I(\cdot)$ refers to an indicator function, $z_{i}$ only takes either 0 or 1 for $i=1,\cdots, k$, $N=\sum_{i=1}^{k}N_{i}$ indicates the total %number of 
population over all $k$ census block groups, and $\sum_{i=1}^{k}z_{i}=1$. A simple random sample is critical for inferring the population, and sampling bias can lead to misleading conclusions, such as those in 2016 U.S. presidential election  \cite{meng2018statistical}.

For a given internet speed $x$, we are interested in the cumulative probability of internet speed for the entire city:
\begin{equation}\label{cum-prob}
    F(x) = \mathbb{P}\left[y\leq x\right].
\end{equation}
Similarly, we define $F_i(\cdot)$ to be the cumulative distribution at the $i$-th census block group, for $i=1,...,k$. The cumulative probability in~\eqref{cum-prob} is a weighted sum of the regional cumulative probability by the corresponding population, as shown in Lemma~\ref{lemma:cum-prob-weight}.

\begin{lemma}\label{lemma:cum-prob-weight}
    The probability in~\eqref{cum-prob} is equal to the weighted sum of $F_{i}(x)$, where $F_{i}(x)=\mathbb{P}\left[y_{i}\leq x\right]$ for $i=1,\cdots,k$, by its population proportion, i.e.
    \begin{equation}
        F(x) = \sum_{i=1}^{k}\frac{N_{i}}{N}F_{i}(x).
        \label{equ:cdf_true}
    \end{equation}
\end{lemma}

\noindent
The proof is given in Appendix~\ref{appendix-subsec:proof-cdf}.  To estimate the cumulative distribution of the internet speed of a city given in~\eqref{cum-prob}, a conventional way is to use all data points to employ the empirical CDF as:
\begin{equation}\label{emp-cdf}
    \hat{F}_b(x) = \frac{1}{n}\sum_{i=1}^{k}\sum_{j=1}^{n_{i}}I(y_{ij}\leq x),
\end{equation}
where $n= \sum_{i=1}^{k}n_{i}$, indicating the total number of sample. The subscript `$b$' denotes bias, as 
 this empirical CDF in~\eqref{emp-cdf} is a weighted sum of regional empirical CDFs by corresponding sample proportions: 
\begin{equation}\label{emp-cdf-ws}
    \hat{F}_b(x) = \sum_{i=1}^{k}\frac{n_{i}}{n}\hat{F}_{i}(x),
\end{equation}
where $\hat{F}_{i}(x) = \frac{1}{n_{i}}\sum_{j=1}^{n_{i}}I(y_{ij}\leq x)$, indicating the simple empirical CDF of $i$-th census block group for $i=1,\cdots,k$.  Note that $\hat{F}_{i}(x)$ converges to the CDF of the census block group $i$ when the sample size goes to infinity, whereas the $ \hat{F}_b(x)$ typically does not converge to the CDF in \eqref{equ:cdf_true}, when the ratio of the sample $n_i/n$ does not converge to the ratio of the population $N_i/N$ at any census block group $i$.  

%The derivation of~\eqref{emp-cdf-ws} is provided in Appendix. 

%To explain the potential bias in the estimation of~\eqref{cum-prob}, we add an assumption that the heterogeneity of internet speed distributions across different regions.

%\begin{assumption}\label{assum-hetero}
%    For any given $x>0$, the true cumulative probability $\mathbb{P}\left[y_{1}\leq x\right]$, $\mathbb{P}\left[y_{2}\leq x\right]$, $\cdots,\ \mathbb{P}\left[y_{k}\leq x\right]$ are distinct.
%\end{assumption}
%Then a bias in the estimation of~\eqref{cum-prob} is owing to the fact that the distribution of the mixing variable in~\eqref{regional-mixture} depends only through the the population proportions. It follows that the empirical CDF in~\eqref{emp-cdf} is biased for~\eqref{cum-prob}.

\subsubsection{Re-weighing the sample for bias correction}\label{subsubsec:re-weighing}
%Considering the fact that the regional empirical CDF $\hat{F}_{i}(x)$ is unbiased for the regional CDF $F_{i}(x)$, 
A direct approach to correct the regional sampling bias is to construct an unbiased estimator for~\eqref{cum-prob} simply by re-weighing the regional empirical CDFs. The revision of weights shall be based on the population proportions so that the unbiased estimator $\hat{F}_{u}$ for~\eqref{cum-prob} is given by:
\begin{equation}\label{eq-cdf-reweight}
    \hat{F}_{u}(x) = \sum_{i=1}^{k}\frac{N_{i}}{N}\hat{F}_{i}(x).
\end{equation}
One can show that $\mathbb{E}\left[\hat{F}_{u}(x)\right]  = \sum_{i=1}^{k}\frac{N_{i}}{N}F_{i}(x)=F(x)$, meaning that  $    \hat{F}_{u}(x)$ is the unbiased estimator of the CDF for the entire city when the samples at each census block group are simple random samples.
The proof of unbiasedness of  Equation~\eqref{eq-cdf-reweight} is provided in Appendix~\ref{appendix-subsec:unbiased}.

\subsubsection{Re-sampling for bias correction}\label{subsubsec:re-sampling}
In some scenarios, having an unbiased or representative sample is important.
%For instance, an unbiased samples can be used to estimate other statistics, such as variabi
%The re-weighing method in Section~\ref{subsubsec:re-weighing} provides an exact unbiased estimator when samples at each census blocks are simple random samples. But it does not recover a regional-bias-free set of data in order to examine the impact of regional bias on the actual data analysis. As such, 
Thus we introduce a re-sampling approach to provide a set of unbiased samples by correcting the regional sampling bias in the original sample. For the $i$-th census block group, we determine the expected number $n_{i}^{*}$ in the re-sampled data to be:
\begin{equation}\label{eq-def-n-star}
    n_{i}^{*}= \left[\frac{nN_{i}}{N}\right],\ \mbox{ for } i=1\cdots,k,
\end{equation}
where $[x]$ indicates the closest integer of a given real number $x$. In other words, we sub-sample 
%rearrange 
the number of samples at each census block group %as many as it should have been according to 
 based on the census block group's proportion of the population. Then, for the $i$-th census block group, 
%region,
we draw a random sample of internet speed measurement with replacement as many as $n_{i}^{*}$ times. Denote $y_{i1}^{r},\cdots, y_{in_{i}^{*}}^{r}$ as the re-sampled internet speed measurements at $i$-th census block group for $i=1,\cdots,k$. Then, one can define the estimated CDF at a given internet speed $x$ from the re-sampled data as: 

\begin{equation}
    \begin{aligned}
        \hat{F}^{*}_u(x) & = \sum_{i=1}^{k}\sum_{j=1}^{n_{i}^{*}}I\left(y_{ij}^{r}\leq x\right)\\
                        & = \sum_{i=1}^{k}\frac{n_{i}^{*}}{n^{*}}\hat{F}^{*}_{i}(x), \label{eq-cdf-resampled}
    \end{aligned}
\end{equation}
where $n^{*}= \sum_{i=1}^{k}n_{i}^{*}$ and $\hat{F}^{*}_{i}(x)= \sum_{j=1}^{n_{i}^{*}}I\left(y_{ij}^{r}\leq x\right)$. %In the sense that $\hat{F}_{i}^{*}(x)$ still remains unbiased for $F_{i}(x)$ and $n_{i}^{*}/n^{*}\approx N_{i}/N$, 
%The estimator in~\eqref{eq-cdf-resampled} is a desirable approximation for~\eqref{eq-cdf-reweight}. 
The following lemma shows that both the re-sampled estimator in~\eqref{eq-cdf-resampled} and re-weighed estimator in~\eqref{eq-cdf-reweight} can converge to the true CDF of the internet speed in a city when the sample size goes to infinity. 
%asymptotically equivalent as the sample size $n_{i}$ of each region grows to infinity for $i=1,\cdots,k$.

\begin{lemma}\label{lemma::asymptotic}
    Assume that Assumptions~\ref{assum-1} and~\ref{assum-2} hold and that the number $k$ of census block groups and population $N_{i}$ for each $i=1,\cdots,k$ remain fixed as constants. Then, for a given $x\in\mathbb{R}$, we have
    \begin{equation}
        \begin{aligned}
            \hat{F}_{u}(x) & \xrightarrow[]{\mathbb{P}}F(x) \text{ as well as}\\
            \hat{F}^{*}_u(x) & \xrightarrow[]{\mathbb{P}}F(x),
        \end{aligned}
    \end{equation}
    as $n_{i}\rightarrow \infty$ for each $i=1,\cdots,k$ so that $n\rightarrow\infty$ with $\xrightarrow[]{\mathbb{P}}$ denoting convergence in probability.
\end{lemma}
Lemma \ref{lemma::asymptotic} is proved in Appendix~\ref{appendix-asymptotic}. The assessment of the uncertainty of the different methods for estimating the CDF is provided in Appendix~\ref{appendix-subsec::ci}. 
According to Lemma~\ref{lemma::asymptotic}, as long as Assumptions~\ref{assum-1} and~\ref{assum-2} hold, the re-weighted estimator in~\eqref{eq-cdf-reweight} and the re-sampled estimator in~\eqref{eq-cdf-resampled} are both consistent for estimating the CDF of internet speed of a city with a sufficiently large sample for each census block group. This enables us to recover a set of data with regional sampling bias adjusted so that we can easily generate reference results from various sorts of analysis by applying the same analytic tactics to re-sampled data.

%Based on this observation, it is not clear that the inference about the distribution of internet speed is affected by the regional bias found in Section~\ref{subsec:sample_balance}.

%As we see from zoom-in plots in Fig.~\ref{fig-cdf-compare}, the results from three methods are not perfectly identical, but the difference is not remarkable. Rather, for any combination of cities and device types, the graphs from three methods look similar. Although we find evidences for regional bias in both cities and device types, it does not affect the estimation of the marginal distribution of internet speed. 

\begin{figure}[h]
  \centering
  \includegraphics[width=\linewidth]{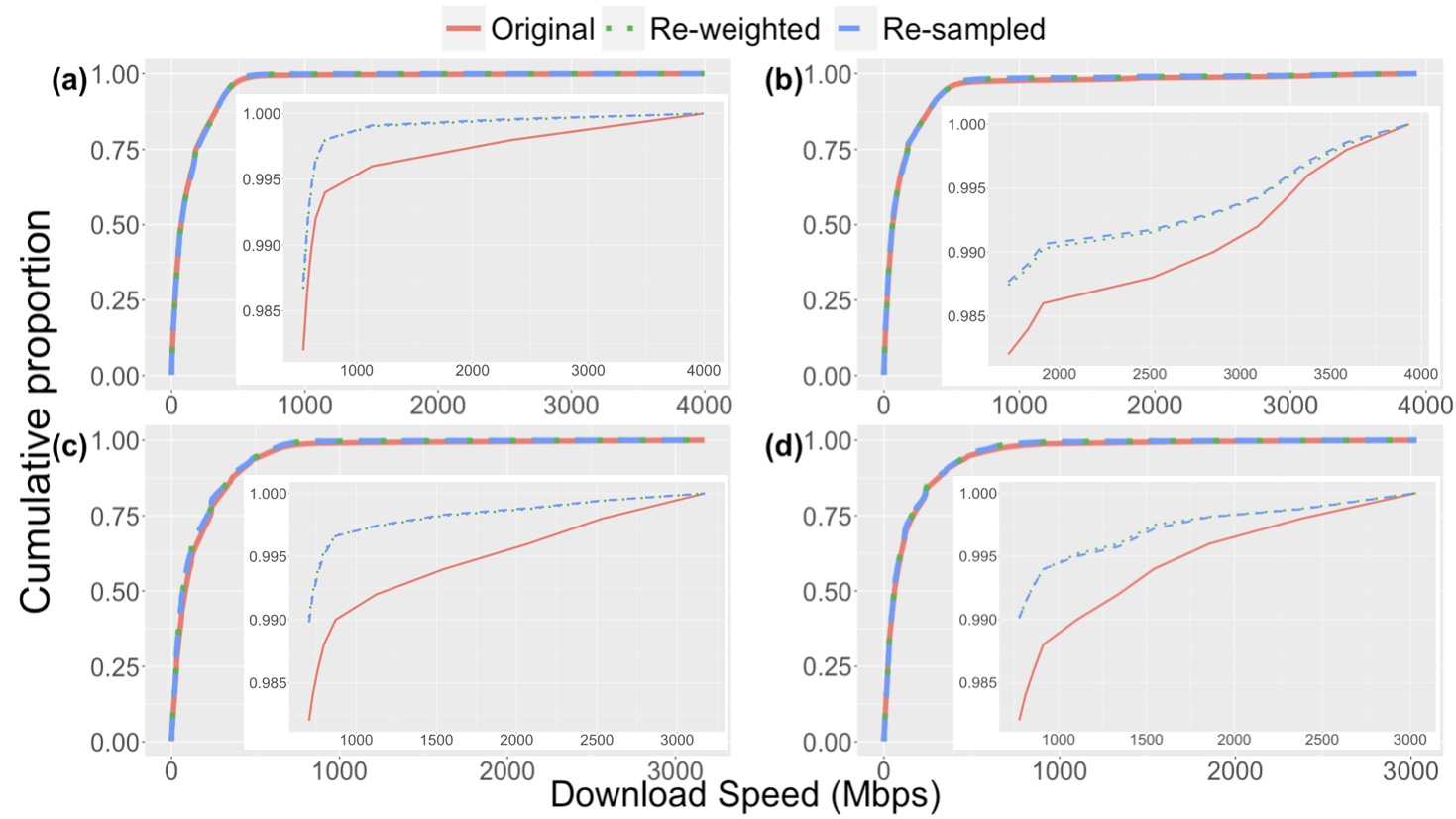}
  \caption{\label{fig-cdf-compare}
  Comparison of empirical CDF from the original samples with re-weighted empirical CDF from original samples and empirical CDF from re-sampled data. (a) iOS devices in City A; (b) Android devices in City A; (c) iOS devices in City B; and (d) Android devices in City B. The insets show zoom-in plots of the estimated CDFs of the internet download speed above the 98$^{th}$ percentile from the original samples.}
\end{figure}

In Fig.~\ref{fig-cdf-compare}, we compare three different CDFs: the empirical CDF from the original samples in \eqref{emp-cdf-ws}, the re-weighted empirical CDF in \eqref{eq-cdf-reweight}, and the empirical CDF from the re-sampled data in \eqref{eq-cdf-resampled}. %For Ookla data from both device types in two cities, 
We find that the estimation from the three methods is not identical, but the differences are much smaller than the discrepancy found from regional sampling bias, as shown in Fig \ref{fig-pop-sam}.  This result indicates that 
even though the Ookla data contains substantial regional sampling bias among census block groups, such bias does not have a large impact on estimating the overall distribution of the internet speed in a city.  

\subsection{Association of regional bias with demographic variables}\label{subsec:boxplot}

While the impact of regional sampling bias over different census block groups on the estimation of the CDF is not as pronounced as the regional bias itself, this impact depends on whether the samples within the census block groups are representative and independent, as outlined in Assumptions~\ref{assum-1}-\ref{assum-2}. 
 These assumptions may not strictly hold in practice.  For instance, if a certain demographic group is over-sampled,  the samples may not be representative in each census block group, violating Assumption \ref{assum-2}. 
%it can serve as a significant indicator of potential distortions in future data analyses. 
Thus, it is crucial to discover potential sampling bias among other sub-groups. %In this regard, it is essential to investigate whether the observed regional bias can be attributed to demographic variables. 
We collect %data %from census bureau 
the demographic variable data for each census block group for the following variables: income, age, gender, educational level, internet penetration rate, and ethnicity. 
The variables that characterize the demographic profile of the $i$-th census block group, for each $i=1,\cdots,k$, are summarized in Table~\ref{tab:demo-list}.

\begin{table*}
\centering
  \caption{List of demographic variables.}
  \label{tab:demo-list}
  \begin{tabular}{c|l}
    \toprule
      Notation & Description \\
    \midrule
    $x_{1i}$ & median household income \\
    $x_{2i}$ & median age  \\
    $x_{3i}$ & \%  of male people \\
    $x_{4i}$ & \% of people with bachelor's or higher degree \\
    $x_{5i}$ & \%  of household with internet subscription plans \\
    $x_{6i}$ & \% of people who are identified as white\\
    $x_{7i}$ & \%  of people who are identified as black\\
    $x_{8i}$ & \% of people who are identified as Asian\\ 
    $x_{9i}$ & \% of people who are identified as Hispanic\\
    \bottomrule
  \end{tabular}
\end{table*}

%\begin{table*}
%\centering
%  \caption{List of demographic variables.}
%  \label{tab:demo-list}
%  \begin{tabular}{cl|cl}
%    \toprule
%      Notation & Description & Notation & Description \\
%    \midrule
%    $x_{1i}$ & median household income & $x_{6i}$ & \% of people who are identified as white\\
%    $x_{2i}$ & median age & $x_{7i}$ & \%  of people who are identified as black\\
%    $x_{3i}$ & \%  of male people & $x_{8i}$ & \% of people who are identified as Asian\\
%    $x_{4i}$ & \% of people with bachelor's or higher degree & 
%    $x_{9i}$ & \% of people who are identified as Hispanic\\
    
 %   $x_{5i}$ & \%  of household with internet subscription plans & & \\
 %   \bottomrule
%  \end{tabular}
%\end{table*}

%\begin{itemize}
%    \item $x_{1i}$: median household income ; 
%    \item $x_{2i}$: median age;
%    \item $x_{3i}$: percentage  of male people;
%    \item $x_{4i}$: percentage of people with bachelor's or higher degree;
%    \item $x_{5i}$: percentage  of household with internet subscription plans;
%    \item $x_{6i}$: percentage of people who are identified as white;
%    \item $x_{7i}$: percentage  of people who are identified as black; 
%     \item $x_{8i}$: percentage of people who are identified as Asian; 
%    \item $x_{9i}$: percentage of people who are identified as Hispanic.
%\end{itemize}

\subsubsection{Characterization and comparison of regions}\label{subsubsec:area-classification}
For each demographic feature $x_{hi}$ where $h=1,\cdots,9$, we have $k$ corresponding values with the index set $\{1,2,3,\cdots,k\}$ as we collect samples from $k$ different census block groups. For each index $i$ where $i=1,\cdots,k$, we classify the census block group into two mutually exclusive groups:
\begin{itemize}
    \item $i$-th census block group is called \textit{over-sampled} if $n_{i}>n^{*}_{i}$; and
    \item $i$-th census block group is called \textit{under-sampled} if $n_{i}<n^{*}_{i}$.
\end{itemize}
We note that there is no census block group where $n_{i}=n_{i}^{*}$ 
%i.e. either $n_{i}<n^{*}_{i}$ of $n_{i}>n_{i}^{*}$ 
in both City A and B, for $i=1,\cdots, k$. Let $\{o_{1},o_{2},\cdots,o_{k_{o}}\}$ and $\{u_{1},u_{2},\cdots,u_{k_{u}}\}$ denote the set of indices for over-/under-sampled regions, with `$o$' and '$u$' denoting over-sampled and under-sampled, respectively, such that $\{o_{1},o_{2},\cdots,o_{k_{o}}\}\cup \{u_{1},u_{2},\cdots,u_{k_{u}}\}=\{1,\cdots,k\}$,  and $\{o_{1},o_{2},\cdots,o_{k_{o}}\}\cap \{u_{1},u_{2},\cdots,u_{k_{u}}\}=\phi$, and $k_{o}+k_{u}= k$. Accordingly, for the $h$-th demographic feature with $h=1,\cdots ,9$, we have two separate groups, denoted as $\{x_{ho_{1}},x_{ho_{2}},\cdots, x_{ho_{k_o}}\}$ and $\{x_{hu_{1}},x_{hu_{2}},\cdots, x_{hu_{k_u}}\}$ for over-/under-sampled census block groups, respectively. 

%comment out as they were discussed later
%We compare the distributions of those groups by box-plots. 
For the $h$-th demographic variable with $h=1,\cdots,9$, define $\bar{x}_{ho} = \frac{1}{k_{o}}\sum_{i=1}^{k_o}x_{ho_{i}}$, and $\bar{x}_{hu} = \frac{1}{k_{u}}\sum_{i=1}^{k_u}x_{hu_{i}}$, which represent the sample mean of the over-/under-sampled regions, respectively. 
To evaluate the difference of location between two distributions, we apply a two-sample t-test with the  test statistics $T$ given by:
%%two sample t test is well-known as it is in stat 101
%Before the application, we assume that, for any $h=1,\cdots,9$, $x_{h1},\cdots, x_{hk}$ are independent, and that, for given $h$, $x_{ho_{1}},x_{ho_{2}},\cdots, x_{ho_{k_o}}$ follow Gaussian distribution with mean $\mu_{ho}$ and variance $\sigma_{h}^{2}$, and $x_{hu_{1}},x_{hu_{2}},\cdots, x_{hu_{k_u}}$ follow Gaussian distribution with mean $\mu_{hu}$ and variance $\sigma_{h}^{2}$. 
%It follows that the test statistics $U$ is given by
\begin{equation}\label{t-stat}
   % U 
    T= \frac{\bar{x}_{ho}-\bar{x}_{hu}}{\sqrt{s^{2}\left(1/k_{o}+1/k_{u}\right)}},
    %\sim t_{k_{o}+k_{u}-2},
\end{equation}
where $s^{2}= \left(\sum_{i=1}^{k_o}(x_{ho_{i}}-\bar{x}_{ho})^{2}+\sum_{i=1}^{k_u}(x_{hu_{i}}-\bar{x}_{hu})^2\right)/(k-2)$. Under the null hypothesis where the distributions of over-/under-sampled census block groups have the same location, the statistic $T \sim t_{k-2} $, a Student's t-distribution~\cite{Casella2002} with $k-2$ degrees of freedom. 
%(which is denoted by a random variable $t_{k-2}$). 

%At 95\% significance level, if the test statistic $u$ (realized value of $U$) is greater than $t_{k-2, 0.975}$ or less than $-t_{k-2, 0.975}$ with $t_{k-2,\alpha}$ being a number such that $\mathbb{P}(U< t_{k-2,\alpha})= \alpha$, then we reject the null hypothesis, which indicates that the population mean values of two groups, $\mu_{ho}$ and $\mu_{hu}$, are distinct. This rejection criterion is equivalent to rejecting the null hypothesis if the p-value is less than 0.05.

\begin{figure}[t]
  \centering
  \includegraphics[width=.9\linewidth]{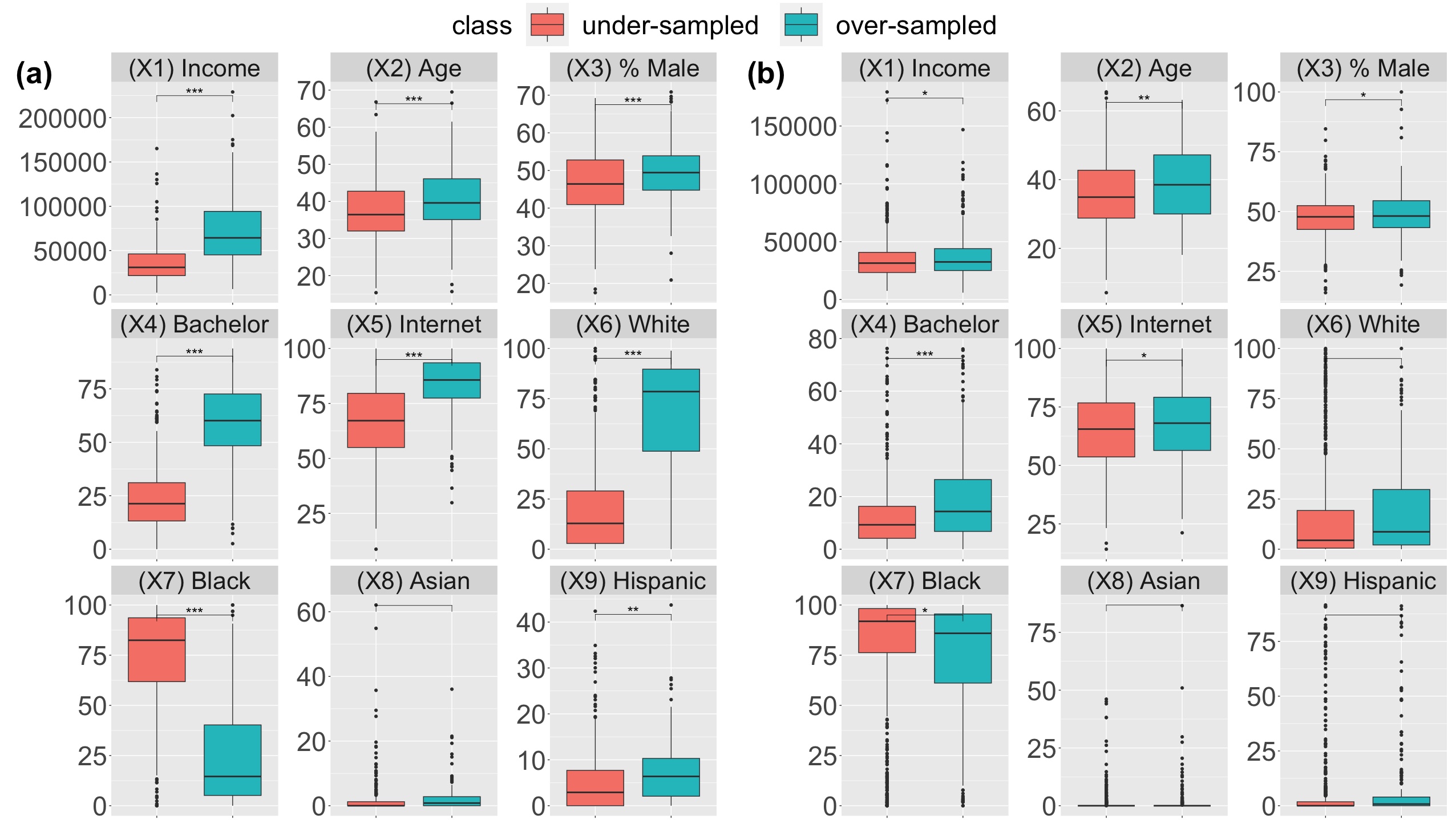}
  \caption{\label{fig-boxplot}
  Comparison of distributions of demographic variables between over-/under-sampled census block groups. Significance is based on the p-values from the two-sample t-test: *** (p-value$<$0.001); ** (p-value$<$0.01); * (p-value$<$0.05); and empty mark for non-significant cases. (a) iOS devices from City A; (b) iOS devices from City B. 
  %The comparison for Android devices is provided in Fig.~\ref{appendix-fig-boxplot-android} of Appendix~\ref{appendix-subsec:bias-detection}.
  }
\end{figure}

\begin{figure}[t!]
  \centering
  \includegraphics[width=.9\linewidth]{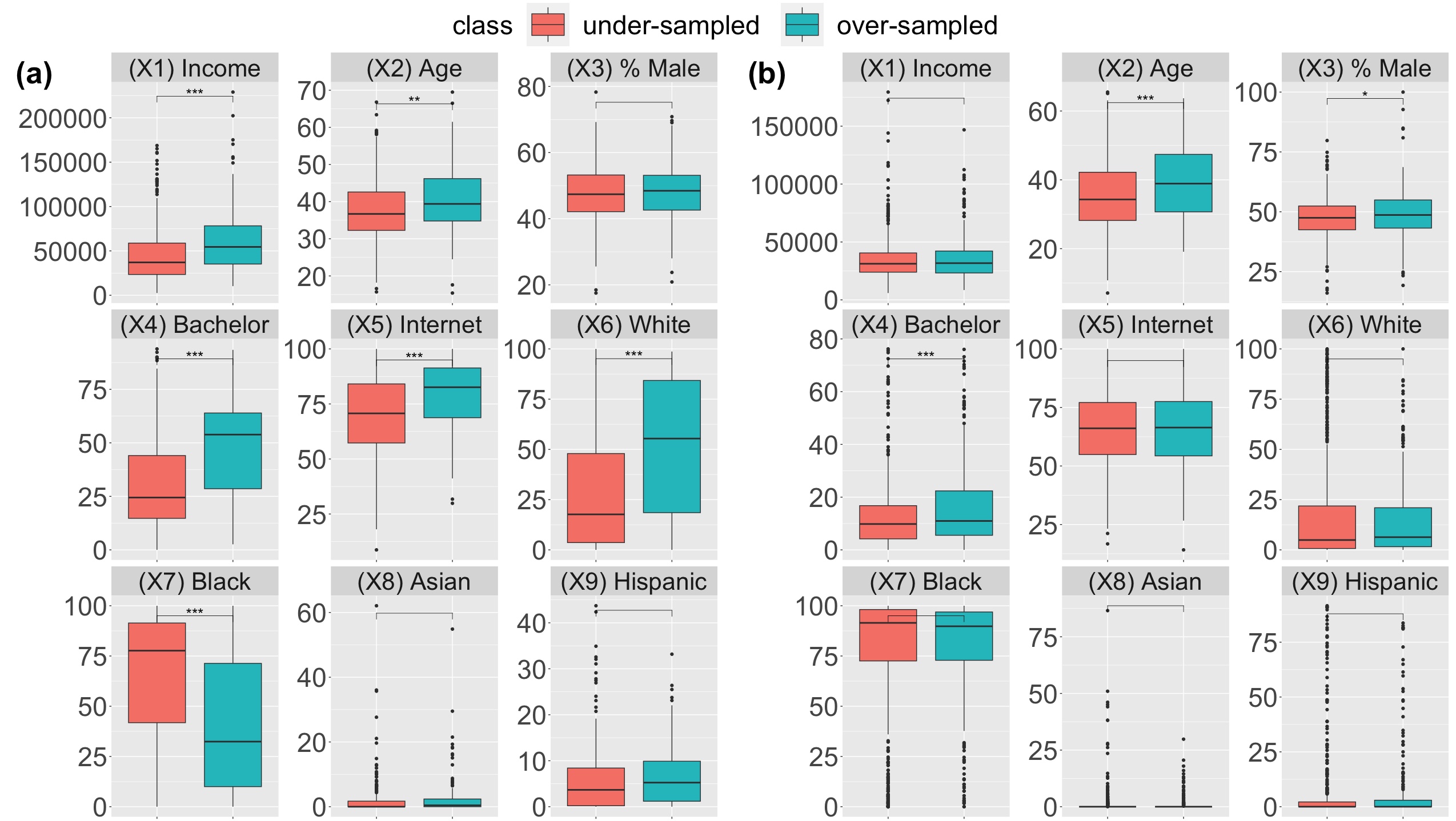}
  \caption{\label{appendix-fig-boxplot-android}
  Comparison of  distributions of demographic variables between over-/under-sampled census block groups. Significance is based on the p-values from the two-sample t-test: *** (p-value$<$0.001); ** (p-value$<$0.01); * (p-value$<$0.05); and empty mark for non-significant cases. (a) Android devices from City A; (b) Android devices from City B. }
\end{figure}

In Fig.~\ref{fig-boxplot} and~\ref{appendix-fig-boxplot-android}, we compare the distributions of demographic variables between over-sampled and under-sampled regions by boxplots, where the stars represent the significance level from the two-sample t-test. In both cities, we find that the sampling of Ookla's Speedtest is indeed inhomogeneous over the demographic variables. Specifically, we observe in both cities that over-sampled census block groups, for either iOS or Android devices, tend to have a greater age and a larger proportion of individuals with bachelor's degrees or higher. However, City A (part (a) in Fig.~\ref{fig-boxplot} and~\ref{appendix-fig-boxplot-android}) presents more prominent asymmetry than City B in many aspects. First, we find in City A that the over-sampled census block groups dominantly have higher income, a larger proportion of the population with a bachelor's degree or higher, and a greater prevalence of households with internet subscriptions. Furthermore, the comparison of Cities A and B reveals the ethnic variety within and between cities. City A shows the contrast between census block groups with a higher percentage of white or black residents, whereas most of the census block groups in City B have a high representation of black residents. This contrast in City A has a relationship with Speedtest sampling; over-sampled regions tend to have a significantly higher representation of white residents and a lower representation of black residents, which is not the case in City B. Additionally, the over-sampled regions for iOS devices in City A tend to have more Hispanic population whereas the difference is not significant in City B.

%indeed associated with the asymmetric sampling structure.
%for each demographic variable with statistical significance from two-sample t-test. 

%and a greater prevalence of households with internet subscriptions.
%On the contrary, and importantly,  under-sampled census block groups are more likely to have a higher representation of black residents compared to over-sampled census block groups.  This is particularly true  in City A. This underscores a significant disparity in sampling patterns among different ethnic groups.

\section{Correlating internet quality with demographic variables}\label{sec:lm}

%In Section~\ref{subsec:cdf}, our observation suggests that the regional bias, as discussed in Section~\ref{subsec:sample_balance}, does not have a critical impact on the estimation of internet speed distribution.
%However, as we delve into 
Section~\ref{subsec:boxplot} showed that Ookla Speedtests from our cities under study are unevenly sampled across different demographic groups.  This result is important if there is a relationship between internet quality (e.g., speed) and demographic variables.  Hence we now ask the question,  
%demographic features of regions could potentially be responsible for the observed regional bias. 
does internet access quality have a relationship to 
demographic variables?
%depend on demographic variables? 
%For instance, regions with higher median household income tend to yield more samples than expected. In this context, it is essential to recognize that the regional bias may indeed influence the joint distribution of internet speed and demographic features. 
To answer this question, we correlate the internet speed with the demographic profiles in each census block group for both the original data set and the re-sampled data set after correcting the regional sampling bias. Because demographic profiles of individuals who conducted Speedtests are not available, we will integrate the Speedtests and demographic variables at two different spatial granular levels in the next section. 
%This interplay is of paramount importance for regression analyses. 
%To illustrate, when implementing the re-sampling method described in Section~\ref{subsubsec:re-sampling} to rectify the regional bias, one might inadvertently end up with a larger number of internet speed samples associated with regions having lower median income, which can change results from regression analysis of internet speed on demographic features. In this section, we examine the implications of regional bias in the context of linear regression analysis.

\subsection{Multiple linear regression}\label{subsec:lm}
%We briefly introduce multiple linear regression models with shared predictors and their relationship in terms of estimating parameters. 
We begin with a multiple linear regression model that relates internet speed $y_{ij}$, the Speedtest $j$ at census block group $i$, to the demographic variable  $\mathbf x_i$:  
\begin{equation}\label{eq-model-indi}
    y_{ij} = \mathbf x_{i}^{T}\boldsymbol \beta + \epsilon_{ij},
\end{equation}
where $\mathbf x_{i} = (1, x_{1i},x_{2i},\cdots,x_{pi})^{T}$ denotes a $(p+1)$-dimensional vector of predictors (such as income, age, level of education, etc.) with $p=9$, $\boldsymbol \beta = (\beta_{0}, \beta_{1}, \cdots, \beta_{p})^{T}\in\mathbb{R}^{(p+1)}$ is a vector of linear coefficients with $p=9$, and $\epsilon_{ij} \sim \mathcal{N}(0,\sigma^{2})$ denotes a Gaussian white noise with variance $\sigma^2$, for $j = 1,\cdots, n_{i},$ and $i=1,\cdots, k$. Note that the observations from the $i^{th}$ census block group, denoted as $y_{i1},\cdots,y_{in_{i}}$, share the same predictor vector $\mathbf x_{i}$ for $i=1,\cdots,k$. It follows that the model~\eqref{eq-model-indi} is equivalent to: 
\begin{equation}\label{eq-model-indi-matrix}
    \mathbf y = \mathbf V \boldsymbol \beta +\boldsymbol \epsilon,
\end{equation}
where $\mathbf y = (y_{11},\cdots,y_{1n_{1}}, \cdots,y_{k1},\cdots,y_{kn_{k}})^{T}\in\mathbb{R}^{n}$ with $n= \sum_{i=1}^{k}n_{i}$, $\mathbf{V} = [\mathbf x_{1} \mathbf{1}_{n_{1}}^{T}  ,\mathbf x_{2} \mathbf{1}_{n_{2}}^{T},\cdots, \mathbf x_{k} \mathbf{1}_{n_{k}}^{T}]^{T}\in\mathbb{R}^{n\times (p+1)}$, $\mathbf{1}_{d}$ denotes a $d$-dimensional column vector with all elements being 1, and \\ $\boldsymbol \epsilon = (\epsilon_{11},\cdots,\epsilon_{1n_{1}},\cdots,\epsilon_{k1},\cdots,\epsilon_{kn_{k}})^{T}\sim\mathcal{MN}(\mathbf 0, \sigma^{2}\mathbf{I}_{n})$ with $\mathbf I_{n}$ being an $n$-dimensional identity matrix.

 %On the other hand, one can construct an aggregated version of model~\eqref{eq-model-indi}, i.e.
 Due to the two different spatial granular levels in model~\eqref{eq-model-indi}, an aggregated version of model~\eqref{eq-model-indi} can be suggested to equalize the granular levels without affecting the estimation of $\boldsymbol \beta$ in model~\eqref{eq-model-indi}.
 \begin{equation}\label{eq-model-agg}
    \bar{y}_{i} = \mathbf x^{T}_{i}\boldsymbol \beta +\epsilon_{i},
\end{equation}
where $\bar{y}_{i} = \frac{1}{n_{i}}\sum_{j=1}^{n_{i}}y_{ij}$, and $\epsilon_{i}$ independently follows $\mathcal{N}(0,\sigma^{2}/n_{i})$ for $i=1,\cdots,k$. Model~\eqref{eq-model-agg} can be written as a  matrix form: 
\begin{equation}\label{eq-model-agg-matrix}
\bar{\mathbf y} = \mathbf X \boldsymbol \beta +\boldsymbol \epsilon_{\text{agg}},
\end{equation}
where $\bar{\mathbf y} = (\bar{y}_{1},\cdots, \bar{y}_{k})^{T}\in\mathbb{R}^{k}$, $\mathbf X = [\mathbf x_{1}, \cdots,\mathbf x_{k}]^{T}\in\mathbb{R}^{k\times (p+1)}$, $\boldsymbol \epsilon_{\text{agg}} = (\epsilon_{1},\cdots,\epsilon_{k})^{T}\sim\mathcal{MN}(\mathbf{0},\sigma^{2}\mathbf W^{-2})$, and $\mathbf{W}$ is a $k\times k$ diagonal matrix with diagonal entries being $\sqrt{n_{1}},\cdots,\sqrt{n_{k}}$. We have the following lemma that connects the individual level model~\eqref{eq-model-indi} and the aggregated level model~\eqref{eq-model-agg}. 

\begin{lemma}\label{lemma::sufficiency}
    Suppose our predictor vectors are given, i.e., we know the vector $\mathbf x_{i}$ for $i=1,\cdots,k$. Let $l$ and $\bar{l}$ be the natural logarithm of the likelihood of model~\eqref{eq-model-indi} and~\eqref{eq-model-agg}, respectively. Then, we have:
    \begin{equation}\label{eq-lemma::suff}
        l(\boldsymbol \beta, \sigma^{2}) = \bar{l}(\boldsymbol \beta, \sigma^{2}) + c_{\sigma^{2}},
    \end{equation}
    where $\bar{l}(\boldsymbol \beta, \sigma^{2}) = k\log\sqrt{2\pi\sigma^{2}} + \sum_{i=1}^{k}\log \sqrt{n_{i}}-\sum_{i=1}^{k}n_{i}(\bar{y}_{i}-\mathbf x_{i}^{T}\boldsymbol \beta)^{2}/(2\sigma^{2})$ and  $c_{\sigma^{2}} =-(n-k)\log\sqrt{2\pi\sigma^{2}}-\sum_{i=1}^{j}\sqrt{n_{i}}-\sum_{i=1}^{k}
    \sum_{j=1}^{n_{i}}(y_{ij}-\bar{y}_{i})^{2}/(2\sigma^{2})$.
\end{lemma}

Lemma~\ref{lemma::sufficiency}  is derived in Appendix~\ref{appendix-suff-proof}. 
%Lemma~\ref{lemma::sufficiency} provides an aspect of justification of model~\eqref{eq-model-agg} as an alternative model for model~\eqref{eq-model-indi}. 
Equation~\eqref{eq-lemma::suff} in Lemma~\ref{lemma::sufficiency} means that the difference between the log-likelihood of model~\eqref{eq-model-indi} and~\eqref{eq-model-agg} does not depend on the linear coefficient $\boldsymbol  \beta$, and only depends on the variance $\sigma^{2}$. 
%Therefore, if we know the value of $\sigma^{2}$ in advance, the estimation of $\boldsymbol \beta$ through model~\eqref{eq-model-indi} is equivalent. 
%from one through model~\eqref{eq-model-agg}. 
Thus the sufficient statistics of the linear coefficients~\cite{Casella2002} is $(\bar{\mathbf y}, \sigma^2)$ with  $\bar{\mathbf y} = (\bar{y}_{1},\cdots, \bar{y}_{k})^{T}$ being the aggregated data. Note the maximum likelihood estimator of linear coefficients does not depend on  $\sigma^2$, whereas the uncertainty of the estimation depends on the noise variance.  
%%This equivalence is related to the concept of statistical sufficiency~\cite{Casella2002}, according to which the sufficient statistic 
%is the aggregated data $\bar{y} = (\bar{y}_{1},\cdots, \bar{y}_{k})^{T}$.
%exhausts all the necessary information in terms of estimating $\boldsymbol \beta$, provided $\sigma^{2}$ known.

 In practice, however, the variance parameter $\sigma^{2}$ of the noise is also unknown. One may be tempted to use the aggregated model (\ref{eq-model-agg-matrix}) to estimate $\sigma^{2}$. 
 The lemma below indicates that the estimation of the variance of the noise by individual-level data is more efficient than that by the aggregated data.
 %Although the assumption of $\sigma^{2}$ being unknown would not change the estimation of $\boldsymbol \beta$, it shows a significant difference in efficiency with respect to the estimation of $\sigma^{2}$.

\begin{lemma}\label{lemma::efficiency}
    Define $\hat{\sigma}^{2} = \mathbf y^{T} (\mathbf I_{n}-\mathbf J)\mathbf y/(n-p-1)$ and $\hat{\sigma}^{2}_{\text{agg}} = \mathbf y^{T}\mathbf W(\mathbf I_{k}-\mathbf H)\mathbf W\mathbf y/(k-p-1)$ where $\mathbf J = \mathbf V(\mathbf V^{T}\mathbf V)^{-1}\mathbf V^{T}$ and $\mathbf H = \mathbf W\mathbf X(\mathbf X\mathbf W^{2}\mathbf X)^{-1}\mathbf X^{T}\mathbf W$ following the notation from~\eqref{eq-model-indi-matrix} and~\eqref{eq-model-agg-matrix}. Note that both $\hat{\sigma}^{2}$ and $\hat{\sigma}^{2}_{\text{agg}}$ are unbiased for $\sigma^{2}$, i.e. $\mathbb{E}\left[\hat{\sigma}^{2}\right] = \mathbb{E}\left[\hat{\sigma}^{2}_{\text{agg}}\right] = \sigma^{2}$. However,
    $$\text{Var}\left[\hat{\sigma}^{2}\right] = \frac{2\sigma^{4}}{n-p-1}<\frac{2\sigma^{4}}{k-p-1} = \text{Var}\left[\hat{\sigma}^{2}_{\text{agg}}\right],$$
    as long as $n>k$.
\end{lemma}

%Lemma~\ref{lemma::efficiency} highlights a significant disparity between model~\eqref{eq-model-indi} and model~\eqref{eq-model-agg}. When estimating $\sigma^{2}$, it proves to be more efficient to utilize model~\eqref{eq-model-indi} rather than model~\eqref{eq-model-agg}, primarily due to the former's reduced uncertainty in the estimation process. 
%This estimation efficiency gains importance because the precision of the estimation of $\sigma^{2}$ directly impacts our ability to gauge the uncertainty surrounding $\boldsymbol \beta$, our parameter of interest.
Thus, we adopt the formulation of model~\eqref{eq-model-indi} to assess the linear relationship between internet download speed (Mbps) and demographic features.
%%this is discussed before
%As described in Section~\ref{subsec:cdf} and~\ref{subsec:boxplot}, we let $y_{ij}$ represent the internet speed measurement of $j$-th unit at $i$-th census block group for $j=1,\cdots, n_{i},$ and $i=1,\cdots,k$, and let $x_{hi}$ denote the $h$-th demographic variable of $i$-th census block group for $h=1,\cdots,9,$ and $ i=1,\cdots,k$. Specifically, recall that
%\begin{itemize}
 %   \item $x_{1i}$: median household income ; 
  %  \item $x_{2i}$: median age;
   % \item $x_{3i}$: percentage (0-100\%) of male people;
    %\item $x_{4i}$: percentage (0-100\%) of people with bachelor's or higher degree;
    %\item $x_{5i}$: percentage (0-100\%) of household with internet subscription plans;
   % \item $x_{6i}$: percentage (0-100\%) of people who are identified as white;
   % \item $x_{7i}$: percentage (0-100\%) of people who are identified as black; 
  %   \item $x_{8i}$: percentage (0-100\%) of people who are identified as Asian; and
  %  \item $x_{9i}$: percentage (0-100\%) of people who are identified as Hispanic.
%\end{itemize}
%We define 10-dimensional vector $\mathbf x_{i} = (1, x_{1i},x_{2i},\cdots,x_{9i})^{T}$ to be our predictor. We apply this notation to model~\eqref{eq-model-indi} to examine the parameter $\boldsymbol \beta = (\beta_0, \beta_1, \cdots, \beta_9)^{T}\in\mathbb{R}^{10}$. 
We apply the regression approaches for both the original samples and the re-sampled data studied in Section~\ref{subsubsec:re-sampling}. 
%To correct the regional bias in samples, we utilize re-sampling method in Section~\ref{subsubsec:re-sampling}.
Similarly, for the re-sampled data, we construct the linear regression as follows:
%model~\eqref{eq-model-indi} as
\begin{equation}\label{eq-model-resample}
    y_{ij}^{r} = \mathbf x_{i}^{T}\boldsymbol \beta^{r}+\epsilon_{ij}^{r},
\end{equation}
where 
%%defined before
%$\mathbf{x}_{i} = (1, x_{1i},\cdots,x_{pi})^{T}\in\mathbb{R}^{(p+1)}$, 
$\boldsymbol \beta^{r} = (\beta_{0}^{r},\beta_{1}^{r},\cdots, \beta_{p}^{r})^{T}\in\mathbb{R}^{(p+1)}$, and $\epsilon_{ij}^{r}$ independently follows $\mathcal{N}(0, \sigma_{r}^{2})$ for  $j=1,\cdots, n_{i}^{*}$, and $i=1,\cdots, k$.

\subsection{Model selection}\label{subsec::model-selection}

Instead of directly comparing the full models presented in~\eqref{eq-model-indi} and~\eqref{eq-model-resample}, we employ a backward model selection procedure to select demographic variables that significantly impact measured internet speeds based on the Akaike Information Criterion (AIC) ~\cite{akaike1998information}.
%%use original reference of AIC instead of a recent review one
%jieding2018modelselection
Let $\mathcal{M}_{m}$ denote a model with index $m$, and define $d_{m}$ as the dimension of model $\mathcal{M}_{m}$. The AIC for $\mathcal{M}_{m}$ is defined as follows:
\begin{equation}\label{eq-AIC}
    \text{AIC}_{m} = -2\hat{l}_{n,m}+d_{m},
\end{equation}
where $\hat{l}_{n,m}$ represents the maximum log-likelihood of model $\mathcal{M}_{m}$ given $n$ observed data. Commencing with a full model featuring $p$ predictors, 
%%shouldn't be M_m 
%denoted as $\mathcal{M}_{m}$,
we consider $p$ distinct submodels created by eliminating one variable at a time. For each submodel, we compute the corresponding AIC using~\eqref{eq-AIC}. If any submodel exhibits a smaller AIC compared to the existing model, we select that model for the subsequent stage and repeat the same process. If no such submodel is found to have a small AIC, we conclude with the existing model 
%$\mathcal{M}_{m}$ 
as the final selected model.

\begin{figure}[t]
  \centering
  \includegraphics[width=\linewidth]{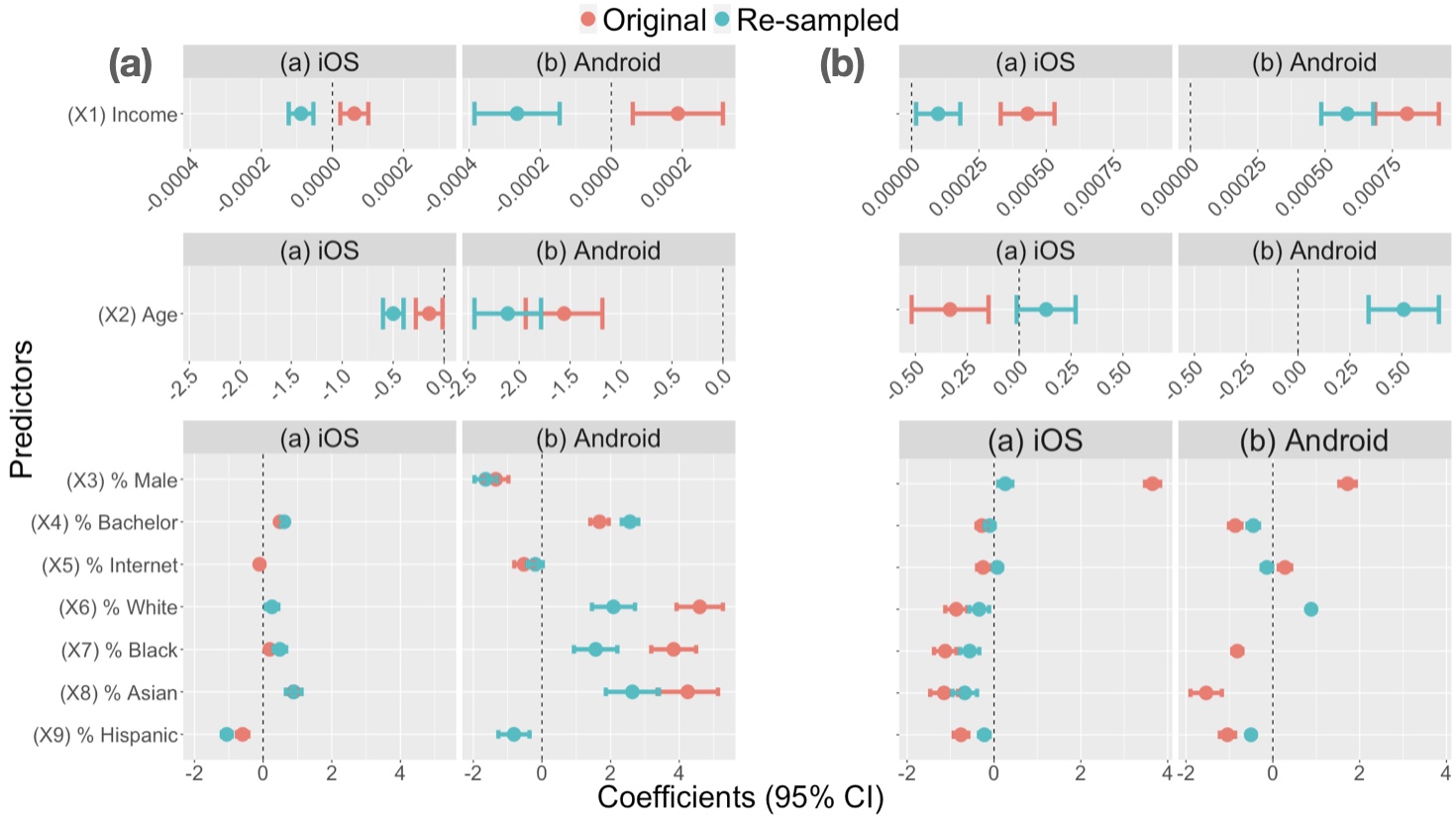}
  \caption{\label{fig-coef-compare}
  Comparison of regression coefficient estimates (dots) and 95\% confidence intervals (bars) from original data and re-sampled data. Multiple linear regression is conducted with backward variable selection by AIC for both iOS and Android from (a) City A, and (b) City B. %95\% confidence interval (CI) is shown with bars. 
  Only the variables selected from the model selection step are shown in the figure.}
\end{figure}

\subsection{Disparity of internet quality among demographic groups}

Fig. \ref{fig-coef-compare} provides estimated linear coefficients from the regression analysis conducted on both the original and re-sampled data after the backward elimination technique using AIC. It presents estimates of coefficients and their corresponding 95\% confidence intervals for variables, whereas the variables dropped from the backward selection procedure are not shown. From the graphs, we can make multiple observations. 

First, there is a significant negative correlation between measured internet speed and age in City A,  City C, and City D shown in Fig \ref{appendix-fig-coef-compare} in the Appendix. This means internet speed is greater in census block groups with a younger population on average, given other estimated variables. City B is an exception, as the effect of age does not seem to be clear. 

Second, regions with a larger percentage of  Hispanic population generally have lower measured internet speed for both cities and device types shown in Fig. \ref{fig-coef-compare}. For the two cities shown in Fig.\ref{appendix-fig-coef-compare} in Appendix~\ref{appendix-subsec:lm}, the effect of the percentage of the Hispanic population is not as clear as the two cities shown in Fig. \ref{fig-coef-compare}; however, after examining the pairwise correlation plot between the covariates in Fig. \ref{appendix-fig-corrplot}-\ref{appendix-fig-corrplot_android_CD}, we find that the percentage of the Hispanic population is negatively correlated with the  percentage of the bachelor's degree and availability of the internet. Thus the effect of Hispanic percentage in population can be partly offset by  these two effects. For instance, the coefficient of bachelor's degree in the re-sampled data is significantly larger than zero in part (a) of Fig. \ref{appendix-fig-coef-compare}, which implicitly suggests that the regions with higher Hispanic percentage may have comparatively lower internet speed, as these regions tend to have a lower percentage of residents with bachelor's degrees. 

Third, we find that the regression coefficients for median income are positive for the original data of both device types in  cities A and B, meaning that the regions with higher income tend to have faster-measured internet speed. This can likely be attributed to the availability of faster internet plans, as well as the higher purchasing power of local residents; prior work found that the median income of the census block groups  play a critical role in determining whether a region gets a fiber deployment and consequently faster internet speeds~\cite{paul2023decoding}. 
The coefficients of median income in the re-sampled data are positive for City B but negative for City A, as shown in Fig.~\ref{fig-coef-compare}. In both cities, the linear coefficients of income in the re-sampled data are smaller than the ones from the original data. To further explore the difference, we find that the pairwise correlation between income and bachelor's degree is strongly positive in both cities, as shown in Fig.~\ref{fig-corrplot}. 
%correlation between income and other variables are plots in Part (a) in Fig. \ref{fig-corrplot}, which shows strong positive correlation between income and bachelor degree.
The estimated linear coefficients of the bachelor's degree in the re-sampled data are larger than those in the original data
for both devices and cities. As  the bachelor's degree and income have strongly positive correlation, the larger coefficient of the bachelor's degree explains the positive impact on the internet speed, which makes the coefficients of the income smaller in the re-sampled data. Note that an estimated coefficient represents the conditional effects of a covariate given all other covariates in multiple linear regression. The effect of a covariate typically depends on the effects of other variables as multicollinearity of the covariates is common in practice \cite{mela2002impact}. 

The regression analysis of Speedtest data from the two cities shows that measured internet speed critically depends on demographic profiles of regions, such as the income, education level and ethnic composition. Future work should study the reasons behind these associations, such as the availability of faster internet plans, and their cost per bit, e.g. carriage value~\cite{paul2023decoding}.

\begin{figure}[t]
  \centering
  \includegraphics[width=.8\linewidth]{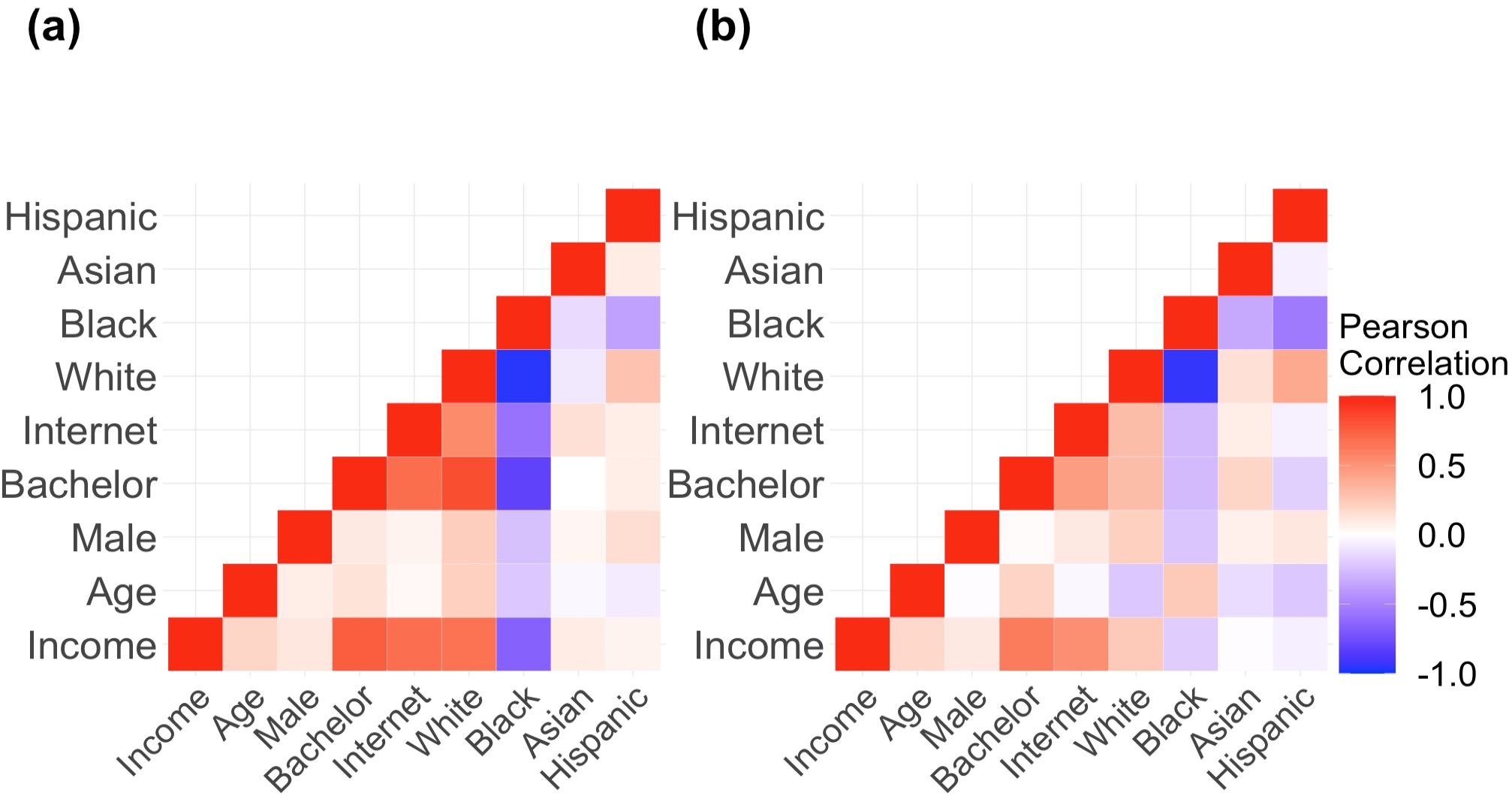}
  \caption{\label{fig-corrplot}
   Heatmap of pair-wise correlation coefficients between demographic covariates from re-sampled data for 
   %from original data. 
   iOS devices in
   (a) City A and (b) City B. 
   %The case of Android devices is provided in Appendix
   }
\end{figure}

\section{Temporal progression of  internet speed}\label{sec:temporal}

Estimating  time-dependent internet speed is important for assessing the change of internet quality over time. In this section, we investigate the temporal trend of measured internet speed using both linear regression analysis and Gaussian processes for modeling time sequences.  We also compare the original samples and the bias-corrected samples to evaluate whether the regional sampling bias affects the estimation of temporal analysis. For both cities, we have data from 05-31-2020 to 12-31-2021. 
%For City A, we have the speed measurements  from 05-31-2020 to 12-31-2021 and for City B, the data from 01-01-2021 to 12-31-2021 are available.

\subsection{Assessing the linear trend of internet speed}\label{subsec:linear-trend}

We first analyze the linear trends of internet speed. Let $y(t_{i})$ denote the measured internet speed at time point $t_{i}$ for $i = 1,\cdots, n$; i.e., we have $n$ distinct time points and each $t_{i}$ corresponds to a positive real number indicating the time lapse from the starting date measured by days.
%how days it lapsed from the date 05-31-2020, 
%measured by seconds. 
The linear regression model of internet speed with an intercept and time as a covariate is given by: 
\begin{equation}\label{eq-time-lm}
    y(t_{i}) = \beta_{0t} + \beta_{1t}t_{i} +  \varepsilon_{i},
\end{equation}
where $\varepsilon_{i}$ represents Gaussian white noise, with variance $\sigma_{\varepsilon}^{2}$ for $i=1,\cdots, n$. We focus on the temporal change rate measured by the linear coefficient over time $\beta_{1t}$, where the maximum likelihood estimator  is the least square estimator below:

\begin{equation}\label{eq-mle-time}
    \hat{\beta}_{1t} = \frac{\sum_{i=1}^{n}(t_{i}-\bar{t})(y(t_{i})-\bar{y})}{\sum_{i=1}^{n}(t_{i}-\bar{t})^{2}},
\end{equation}
where $\bar{t} = \frac{1}{n}\sum_{i=1}^{n}t_{i}$ and $\bar{y} = \frac{1}{n}\sum_{i=1}^{n}y(t_{i})$.

To evaluate the impact of regional sampling bias, we obtain a re-sampled data set with bias-correction introduced in Section~\ref{subsubsec:re-sampling}. Let $t_{i}^{r}$ denote the distinct time point after re-sampling, for $i=1,\cdots, n^{r}$. We define $y(t_{i}^{r})$ as the re-sampled internet speed at time point $t_{i}^{r}$. The linear regression model using the re-sampled data is given by: 
\begin{equation}\label{eq-time-lm-re}
    y(t_{i}^{r}) = \beta_{0t}^{r} + \beta_{1t}^{r}t_{i}^{r} +  \varepsilon_{i}^{r},
\end{equation}
where $\varepsilon_{i}^{r}$ represents a Gaussian white noise with variance $\sigma_{\varepsilon^{r}}^{2}$ for $i=1,\cdots, n^{r}$. The maximum likelihood estimator of $\beta_{1t}^{r}$ in model~\eqref{eq-time-lm-re} follows: 
\begin{equation}\label{eq-mle-time-re}
    \hat{\beta}_{1t}^{r} = \frac{\sum_{i=1}^{n^{r}}(t_{i}^{r}-\bar{t}^{r})(y(t_{i}^{r})-\bar{y}^{r})}{\sum_{i=1}^{n^{r}}(t_{i}^{r}-\bar{t}^{r})^{2}},
\end{equation}
where $\bar{t}^{r} = \frac{1}{n^{r}}\sum_{i=1}^{n^{r}}t_{i}^{r}$ and $\bar{y}^{r} = \frac{1}{n^{r}}\sum_{i=1}^{n^{r}}y(t_{i}^{r})$.
%With regard to the impact of regional bias on analysis of temporal progression of internet speed, we are interested in the comparison of $\hat{\beta}_{1t}$ in~\eqref{eq-mle-time} with $\hat{\beta}_{1t}^{r}$ in~\eqref{eq-mle-time-re}.

%The simple linear regression models of internet speed on time point, as described in~\eqref{eq-time-lm} or in~\eqref{eq-time-lm-re}, can provide condensed snapshots of progression of internet speed, which are summarized by $\hat{\beta}_{1t}$ and $\hat{\beta}_{1t}^{r}$ in~\eqref{eq-mle-time} and~\eqref{eq-mle-time-re}, respectively. 
There are two limitations of the linear regression analysis. First, the estimated linear coefficients in~\eqref{eq-time-lm} and~\eqref{eq-time-lm-re} can only capture average change over a time period. Second, the assumption is that the residuals are independent over time, whereas the Speedtest measurements are temporally correlated. 
%white noises in model~\eqref{eq-time-lm} and~\eqref{eq-time-lm-re} inherently imply independence between observed internet speed measurements, linear regression models does not incorporate any correlation structure within time series of internet speed measurements.
%In the following subsection, we introduce state space models to address correlated speed measurements.
To avoid these limitations, we  introduce Gaussian processes  for modeling the time sequences and accelerate the computation by state space representation without approximation.

\subsection{Modeling the internet speed by state space models}
Internet speeds are temporally correlated. A common way to model the temporal or spatio-temporal data is by Gaussian processes (GPs) \cite{rasmussen2006gaussian}. However, the complexity of computing the likelihood function and making predictions by GPs  increases cubically fast along with the number of observations, due to computing the inversion and log determinant of the covariance matrix.  In our study, the number of measurements for each device in a city is between $10^5$-$10^6$,  which makes directly computing the likelihood  by GPs prohibitively slow. Fortunately, GPs with some widely used covariance functions, such as the Mat{\'e}rn covariance function~\cite{handcock1993bayesian} with half-integer roughness parameters, can be equivalently represented by linear state space models, which makes computational complexity linearly increase with respect to the number of observations without making any approximations \cite{hartikainen2010kalman}.  We briefly introduce a GP model of Speedtest measurements and relate it to the state space model for fast computation for the original Speedtest observations. The fast algorithm through the state space representation can be similarly applied to the regional bias-corrected samples.  

Suppose any internet speed measurement is modeled by a noisy Gaussian process, meaning that any marginal distribution at time $\{t_1,...,t_n\}$ follows a multivariate normal distribution $(y(t_1),...,y(t_n))^T\sim \mathcal{N}(\mathbf 0, \sigma^2 (\mathbf R+\eta \mathbf I_n))$, where $\sigma^2$ and $\eta$ are variance and nugget parameters, respectively, and $\mathbf R$ is a correlation matrix with the $(i,j)$th term parameterized by a kernel function $K(t_i,t_j)$. Denote $d=|t-t'|$ as the distance between any time points $t$ and $t'$. We focus on the Mat{\'e}rn covariance function, which has the expression: 
 \begin{equation}
    \sigma^2 K(d)= \sigma^2 \frac{2^{1-\nu}}{\Gamma(\nu)}\left(\frac{\sqrt{2\nu}d}{\gamma} \right)^{\nu} \mathcal{K}_\nu\left(\frac{\sqrt{2\nu} d}{\gamma} \right), 
 \end{equation}
where $\Gamma(\cdot)$ is the gamma function, $\mathcal{K}_{\nu}(\cdot)$ is the modified Bessel function of the second kind with a positive parameter $\nu$ and $\gamma$ is a range or lengthscale parameter of the correlation. The Mat{\'e}rn covariance has a closed form expression when the roughness parameter is a half-integer, $\nu=\frac{2m+1}{2}$ for $m\in \mathbf N$. For instance, the Mat{\'e}rn with  $\nu=5/2$ has the expression:
\begin{equation}
\begin{aligned}
  \sigma^{2} K(d) = \sigma^{2} \left(1+\frac{\sqrt{5} d}{\gamma}+ \frac{5 d^2}{3 \gamma^2}\right) \exp \left(-\frac{\sqrt{5} d}{\gamma}\right).
    \label{equ:matern_5_2_kernel}
\end{aligned}
\end{equation} 
An appealing feature of the Mat{\'e}rn covariance is that the process is  $\lfloor\nu-1 \rfloor$ mean squared differentiable \cite{Gu2018robustness}, as the smoothness of the process can be controlled by the roughness parameter. 

Suppose we have internet speed measurements at $n$ time points, denoted as $\mathbf y=(y(t_1),...,y(t_n))^T$. A conventional method is to estimate the parameter by the maximum likelihood estimator. Note that given range parameter and nugget parameters $(\gamma,\eta)$, the maximum likelihood estimator for variance is $\hat \sigma^2=S^2/n$ where $S^2=\mathbf y^T \mathbf{\tilde R}^{-1}\mathbf y$ with   $ \mathbf{\tilde R}=\mathbf R+\eta \mathbf I_n$. Plugging the  $\hat \sigma^2$ into the likelihood function, the profile likelihood \cite{Gu2018robustness}   follows: 
\begin{equation}
 p(\mathbf y \mid  \gamma, \eta,\hat \sigma^2)\propto |\mathbf{\tilde R}|^{-1/2}|S^2|^{-n/2}.
 \label{equ:mle}
 \end{equation}
The parameters can be obtained by maximizing the log profile likelihood: $(\hat \gamma, \hat \eta)=\mbox{argmax}_{ \gamma,\eta} \mbox{log}(p(\mathbf y \mid  \gamma, \eta,\hat \sigma^2))$.  After obtaining the MLE,   the predictive distribution at any time point $t$ follows a normal distribution: 
\begin{equation}
 (y(t)\mid \mathbf y, \hat \sigma^2, \hat \eta, \hat \gamma) \sim \mathcal N( \hat y(t), \hat \sigma^2 K^*(t)), 
 \label{equ:pred_dist}
 \end{equation}
 where $\hat y(t)=\mathbf r^T(t) \mathbf{\tilde R}^{-1} \mathbf y$ with $\mathbf r(t)=(K(t,t_1),...,K(t,t_n))^T$ and $K^*(t)=K(t,t)+\eta-\mathbf r^T(t)\mathbf{\tilde R}^{-1} \mathbf r(t)$. The predictive mean $\hat y(t)$ is often used for predicting $y(t)$ and the predictive intervals can be obtained from  (\ref{equ:pred_dist}) for quantifying the uncertainty in prediction. 

\begin{figure}[t]
  \centering
  \includegraphics[width=\linewidth]{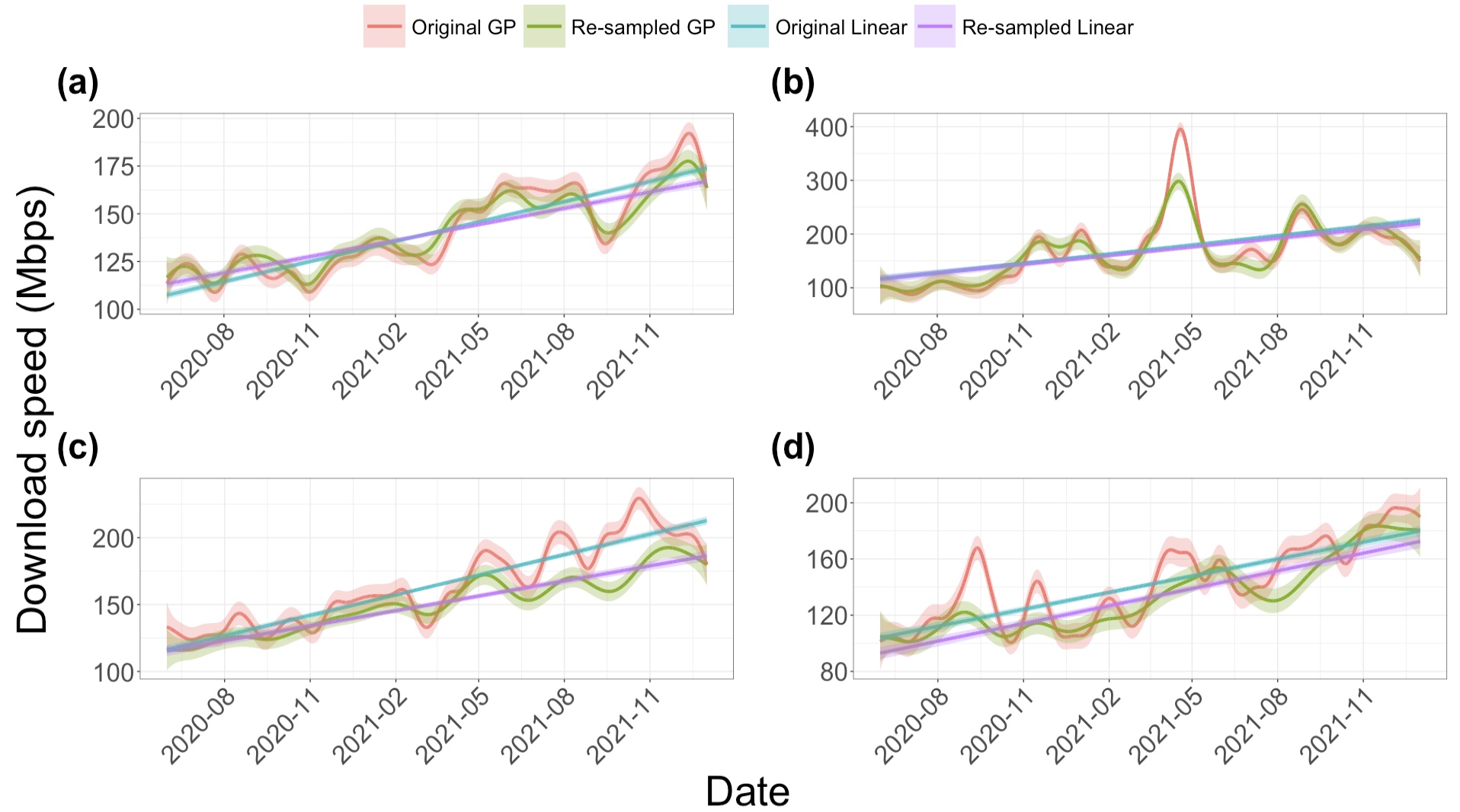}
  \caption{\label{fig-time-trend}
  Comparison of temporal trend of internet download speed from original versus re-sampled data based on linear and Gaussian process regression. (a) iOS devices in City A; (b) Android devices in City A; (c) iOS devices in City B; and (d) Android devices in City B. The solid curves are predictive mean and shaded areas are 95\% percent interval of the estimation.}
\end{figure}

Directly computing the likelihood function or predictive distribution requires inverting an $n\times n$ covariance matrix, which has computational complexity $\mathcal O(n^3)$. Here $n$ can be at the order of $10^6$, prohibiting directly computing GP models. Fortunately, the Mat{\'e}rn covariance with half-integer parameters can be written as a set of stochastic differential equations and the solution follows a continuous-time state space model  \cite{hartikainen2010kalman}. For instance, the GP with a Mat{\'e}rn covariance with roughness 5/2 in (\ref{equ:matern_5_2_kernel}) can be written as a state space model below \cite{gu2020fast}: 
\begin{align}
\label{equ:dlm_matern_5_2}
\begin{split}
y(t_i)&= \mathbf F\bm \theta(t_i) + \epsilon_i, \\
\bm \theta(t_i)&=\mathbf G(t_i) \bm \theta(t_{i-1}) +\mathbf w(t_i), \,    
\end{split}
\end{align}
where $\mathbf F=(1,0,0)$, $\mathbf w(x_i) \sim \mathcal{N}(0, \mathbf W(x_i))$ for $i=2,...,N$, and the initial state follows $\bm \theta(x_1) \sim \mathcal{MN}(\mathbf 0, \mathbf W(x_1))$. The closed-form expression of $\mathbf G(x_i)$ and $\mathbf W(x_i)$ in (\ref{equ:dlm_matern_5_2}) can be found in Appendix A in \cite{gu2022scalable}. 

With the state space representation, we can compute the likelihood function and predictive distribution by Kalman filter \cite{kalman1960new} and RTS smoother \cite{rauch1965maximum}. This algorithm is commonly known as the forward filtering and backward smoothing (FFBS) algorithm \cite{West1997,petris2009dynamic}. The details for computing the likelihood and predictive distribution for state space representation of GP with the covariance in (\ref{equ:matern_5_2_kernel}) are provided in Lemma 2 in \cite{gu2022scalable}. Computing the likelihood in (\ref{equ:mle}) and predictive distribution in (\ref{equ:pred_dist}) using the FFBS algorithm reduces the computational operations from $\mathcal O(n^3)$ to $\mathcal O(n)$ operations without approximation. The computational advance enables us to estimate the nonlinear temporal trend of internet speed with a massive number of crowdsourced  observations.

\subsection{Temporal progression of measured internet speeds}

We plot the estimated temporal progression of  the download speed based on both linear and state space models in Fig.~\ref{fig-time-trend}.
%and compare the estimation from original and re-sampled data for both device types in City A and B. 
 Both models show that download speeds measured by Speedtest improve over time for both cities and device types. However, the improvements of the download speed do not appear to be homogeneous over time. %Although the time range of the available speedtest data are different for City A and City B, 
 We find a comparatively large improvement occurs at the beginning of 2021 for  iOS devices in both cities, shown in parts (a) and (c) in Fig \ref{fig-time-trend}.  Speed improvement can also be found for  Android devices, as shown in parts (b) and (d); however the variation of the estimation for Android devices seems larger, as the sample size from Android devices is much smaller than iOS devices, particularly for City A, as shown in Table \ref{tab:ookla}. The increase may be partly due to the acceleration of the deployment and marketing of fiber internet since 2020 \cite{broadband2021}, as we find that a high proportion of the Speedtest measurements have substantially faster speeds than others since late 2020, and the  proportion grows over time.

Second, the estimation from re-sampled data suggests that the trend from iOS devices tends to be overestimated (part (a) and (c) in Fig. \ref{fig-time-trend}). This overestimation is larger in City B, particularly during 2021, as both linear and state space models shows the noticeable difference of fitting between original and re-sampled data. We suspect that the overestimation is due to more Speedtests from high-speed internet plans, such as fiber internet plan, as subscribers of these plans may tend to submit more Speedtests to validate the speed from these high-speed internet plans. Re-sampling among census block groups can address a part of this bias, as it samples more from regions with relatively smaller numbers of tests compared to their population, which may have lower-speed internet plans.  The temporal trends of internet download speed for the two other cities are plotted in Figure \ref{appendix-fig-time-trend} in Appendix \ref{appendix-subsec:temporal}. The estimation by both the original and bias-corrected samples shows the improvement of download speeds over time, and the difference between the estimation is not large. 

\begin{table*}[t]
\centering
  \caption{Linear trends of measured internet download speed (Mbps) per day.}
  \label{tab:time-lm}
  \begin{tabular}{cccc}
    \toprule
      City & Device type & \thead{Estimates $\hat{\beta}_{1t}$\\(95\% CI)} & \thead{Estimates $\hat{\beta}_{1t}^{r}$\\(95\% CI)} \\
    \midrule
       \multirow{2}{*}{A} 
       &  iOS & 
       \thead{0.1149 \\ (0.1088, 0.1209)} & 
       \thead{0.0928 \\ (0.0857, 0.0999)}\\
    
          & Android & 
          \thead{0.1879 \\ (0.1701, 0.2058)} &  
          \thead{0.1788 \\ (0.1546, 0.2029)}\\
     \midrule
     \multirow{2}{*}{B} 
     &  iOS & 
     \thead{0.1661 \\ (0.1567, 0.1755)} & 
     \thead{0.1225 \\ (0.1110, 0.134)}\\
    
          & Android & 
          \thead{0.1309 \\ (0.1200, 0.1417)} &  
          \thead{0.1370 \\ (0.1232, 0.1509)}\\
    \bottomrule
  \end{tabular}
\end{table*}

Finally, we compare the estimates of linear coefficients $\beta_{1t}$ and $\beta_{1t}^{r}$ for both iOS and Android devices from City A and B in Table~\ref{tab:time-lm}. For both device types, the estimates are positive, which suggests that the download speed increases over time. 
%The estimated coefficients for City A are smaller than that of City B, partly because the time period of City A is larger, whereas the largest increase of download speed seems to appear between March 2021 to May 2021 for both cities. 
The estimates of $\beta_{1t}^{r}$ from the bias-corrected samples are smaller than those of $\beta_{1t}$ in the first three rows, indicating that the improvement of internet speed may be slightly overestimated by the original data for these two cities. The estimates of $\beta_{1t}^{r}$ are similar to  $\beta_{1t}$ for Android device download speeds in City B, whereas the intercept $\beta_{0t}^{r}$ is smaller than $\beta_{0t}$, as shown in part (d) of Fig.  \ref{fig-time-trend}. The results from Fig.~\ref{fig-time-trend} and Table~\ref{tab:time-lm} consistently suggest that measured internet speed improves over the time range. 
%, whereas overestimation is not large for these two cities. 
%The estimation 

\section{Conclusion}
\label{sec:conclusion}

 In this paper, we  integrated Ookla  Speedtest measurements with regional demographic profiles for analyzing disparities of measured internet quality and the temporal evolution of internet speed. 
 We developed  re-weighing and re-sampling methods to %meticulously examine
 correct the large regional sampling bias across census block groups.  
 %and reevaluate the disparities and temporal evolution of internet speed.
 Through regression analysis of integrated data, we found that census block groups with higher income, younger population, and fewer Hispanic residents tend towards higher measured internet speeds. Furthermore, we discerned an encouraging trend of internet speed improvement through temporal modeling of Speedtest measurements. Nevertheless, it is essential to approach these findings with caution, as they are susceptible to different biases inherent in Ookla Speedtest measurements. We anticipate that our new methods can be applied to different crowdsourced data and we outline a few directions for further study. 

First, while our current investigation primarily concentrates on urban areas, the realm of speed test measurements in rural and sparsely populated subareas within cities remains largely unexplored. Consequently, a comprehensive analysis of internet performance profiles  between urban and rural areas presents an intriguing prospect. The principle challenge in this pursuit is the scarcity or absence of crowdsourced samples in rural areas, rendering accurate statistical inference on these regions difficult. 
%rendering the acquisition of sufficient data difficult. 
%To overcome this obstacle, marketing efforts may be spent to obtain more observations on rural areas. 
However, it is likely that spatial proximity and the similarity of demographic profiles can be used to infer internet speed through interpolation. 
%researchers may contemplate inferring sampling information for speed tests in rural areas through interpolation based on spatial proximity or the similarity of demographic profiles. 
%In cases with a limited number of rural subareas, theoretical variance analysis may be employed to assess the heterogeneity of internet speed across regions.

Second, our study 
%rooted in two exemplar cities,
naturally stimulates further inquiry into internet speed across other cities or states. Given the expansive coverage of crowdsourced measurements across the United States, researchers can leverage data on a larger scale to generalize internet speed characteristics throughout the country.
To accommodate a more diverse range of states and cities, one plausible modeling approach involves incorporating mixed effects into the model, accounting for spatial variations. 
The nature of these mixed effects may vary depending on the hierarchical structure amongst states, counties, and census block groups.
%hierarchical structure of spatial information.
Identifying a suitable metric to define correlations between regions may be challenging, but demographic similarity stands as a viable option to gauge correlated structures. This methodology enables the construction of a more comprehensive model alongside city- and state-specific explanatory terms.

Third, our analysis of the temporal progression of internet speed motivates future studies in this direction. Our study employing state space models reveals a notable degree of volatility in the time series of internet speed, suggesting that the distribution of internet speed comprises multiple heterogeneous groups over time. An essential latent factor in this context is the internet subscription plan. To address this, one may consider employing a mixture of Gaussian processes or state space models to analyze the temporal evolution of internet speed while accounting for different subscription plans.

Lastly, it is essential to consider potential sources of bias other than the sampling bias among census block groups. While our study has identified an association between regional sampling bias and demographic disparities, further investigation is needed due to the constraints against the availability of demographic profiles corresponding to individual speed tests. Additional data for calibrating the model can be obtained by anonymous surveys to address this limitation. The collection of paired data encompassing internet speed measurements and demographic data of the speed test taker can be used for regression analysis to understand whether other forms of bias affect the estimation between internet quality and demographic features.

%%comment acknowledgement out for now for review?
%\begin{acks}
%%To Robert, for the bagels and explaining CMYK and color spaces.
%\end{acks}
\bibliographystyle{acm}
\bibliography{References_chronical_2022}
%
%\printbibliography

\appendix

\section{Proofs and derivations}
\subsection{Proof of Lemma~\ref{lemma:cum-prob-weight}}\label{appendix-subsec:proof-cdf}
\begin{proof}
Consider a fixed $x\in\mathbb{R}$.
Note that the indicator function $I(y_i\leq x)$ follows Bernoulli distribution with success probability $\mathbb{P}(y_i\leq x)$, for $i=1,\cdots,k$. Based on the law of total expectation, 
%the fact that $I(y_{i}\leq x)$ independently follows Bernoulli distribution with success probability $\mathbb{P}(y_{i}\leq x)$ for $i=1,\cdots,k$, 
we have:
\begin{equation*}
\begin{aligned}
    \mathbb{P}(y\leq x) & = \mathbb{E}\left[I(y\leq x)\right]\\
    & = \mathbb{E}\left[\mathbb{E}\left[I(y\leq x)|\mathbf z\right]\right]\\ %\quad\text{(tower property)}\\
    & = \sum_{i=1}^{k}\mathbb{P}(z_{i}=1)\mathbb{E}\left[I(y\leq x)|z_{i}=1\right]\\
    & = \sum_{i=1}^{k}\mathbb{P}(z_{i}=1)\mathbb{E}\left[I(y_{i}\leq x)\right]\\ %\quad\left(\because y =\sum_{i=1}^{k}y_{i}I(z_{i}=1)\right)\\
    & = \sum_{i=1}^{k}\frac{N_{i}}{N}\mathbb{E}\left[I(y_{i}\leq x)\right]\\ %\quad\left(\because \mathbb{P}(z_{i}=1)=\frac{N_{i}}{N}\right)\\
    & = \sum_{i=1}^{k}\frac{N_{i}}{N}\mathbb{P}(y_{i}\leq x)\\ %\quad\left(\because I(y_{i}\leq x)\sim\text{Bernoulli}(\mathbb{P}(y_{i}\leq x))\right)\\
    & = \sum_{i=1}^{k}\frac{N_{i}}{N}F_{i}(x), %\quad\text{(by definition of }F_{i}(x)).
\end{aligned}
\end{equation*}
where  $\mathbf z$ follows a multinomial  distribution in Equation (\ref{regional-mixture}). 
\end{proof}

\subsection{Proof of unbiased estimator in  Equation~\eqref{eq-cdf-reweight}}\label{appendix-subsec:unbiased}
\begin{proof}
Consider a fixed $x\in\mathbb{R}$. Note that $\hat{F}_{i}(x)$ is unbiased for $F_{i}(x)$  for any census block group $i$, $I(y_{ij}\leq x)$ represents independent Bernoulli process with probability $F_{i}(x) = \mathbb{P}(y_{i}\leq x)$, i.e.
\begin{equation*}
%\label{appendix-eq-unb-region}
\begin{aligned}
    \mathbb{E}\left[\hat{F}_{i}(x)\right] & = \mathbb{E}\left[\frac{1}{n_{i}}\sum_{j}^{n_{i}}I\left(y_{ij}\leq x\right)\right]\\
    & = \frac{1}{n_{i}}\sum_{j=1}^{n_{i}}\mathbb{E}\left[I\left(y_{ij}\leq x\right)\right]\\ %\quad(\text{linearity of expectation operator $\mathbb{E}[\cdot])$})\\
    & = \frac{1}{n_{i}}\sum_{j=1}^{n_{i}}F_{i}(x)\\ %\quad(\text{for given }i,\ I\left(y_{ij}\leq x\right)\sim\text{Bernoulli}(F_{i}(x)),\text{ for } j=1,\cdots,n_{i})\\
    & = \frac{1}{n_{i}}n_{i}F_{i}(x) = F_{i}(x).
\end{aligned}
\end{equation*}
 Based on the linearity of expectation operator $\mathbb{E}[\cdot]$,
we then have: 
%and ~\eqref{appendix-eq-unb-region} enables us to derive
\begin{equation*}
    \mathbb{E}\left[\hat{F}_{u}(x)\right] = \sum_{i=1}^{k}\frac{N_{i}}{N}\mathbb{E}\left[\hat{F}_{i}(x)\right]
    = \sum_{i=1}^{k}\frac{N_{i}}{N}F_{i}(x)=F(x).
\end{equation*}

\end{proof}

\subsection{Proof of Lemma~\ref{lemma::asymptotic}}\label{appendix-asymptotic}
\begin{proof}
    
    By the weak law of large numbers,  we have $\hat{F}_{i}(x)\xrightarrow[]{\mathbb{P}}F_{i}(x)$ and $\hat{F}^{*}_{i}(x)\xrightarrow[]{\mathbb{P}}F_{i}(x)$ for  any real number $x$,  when  $n_{i}\rightarrow\infty$ for $i=1,\cdots,k$. 
    %because of  Assumptions~\ref{assum-1} and~\ref{assum-2}. 
    Since $k$ is a finite number, 
    %$$N_{i}$ is a finite number for all $i=1,\cdots,k$, 
    it follows from Slutsky's Theorem that: 
    $$\hat{F}_{u}(x) = \sum_{i=1}^{k}\frac{N_{i}}{N}\hat{F}_{i}(x)\xrightarrow[]{\mathbb{P}}\sum_{i=1}^{k}\frac{N_{i}}{N}F_{i}(x) = F(x), $$
   when $n_{i}\rightarrow\infty$ at each $i$.  It suffices to show that: 
   \begin{equation}\label{appendix-eq-asym-prop}
      \lim_{n\rightarrow\infty}\frac{n_{i}^{*}}{n^{*}}=\frac{N_{i}}{N}. 
   \end{equation}
    By definition of $n_{i}$ in~\eqref{eq-def-n-star}, we have:
    \begin{equation}\label{appendix-ineq-bound}
        \begin{aligned}
            \frac{nN_{i}}{N}-\frac{1}{2} & \leq n_{i}^{*} = \left[\frac{nN_{i}}{N}\right]\leq \frac{nN_{i}}{N}+\frac{1}{2},\text{ so}\\
            n-\frac{1}{2}k & \leq n^{*} = \sum_{i=1}^{k}\left[\frac{nN{i}}{N}\right]\leq n+\frac{1}{2}k.
        \end{aligned}
    \end{equation}
    Then, from~\eqref{appendix-ineq-bound}, we obtain: 
    \begin{equation}\label{appendix-approx-bound}
        \frac{nN_{i}/N-1/2}{n+k/2} \leq \frac{n_{i}^{*}}{n^{*}}\leq \frac{nN_{i}/N+1/2}{n-k/2}.
    \end{equation}
    Letting $n\rightarrow\infty$, we apply the Squeeze Theorem to the inequality in~\eqref{appendix-approx-bound} to yield the result of~\eqref{appendix-eq-asym-prop}.
\end{proof}

\subsection{Derivation of confidence intervals for the methods in Section~\ref{subsec:cdf}}\label{appendix-subsec::ci}
We derive the confidence intervals of different ways for estimating the CDF. For all methods, we denote $x\in\mathbb{R}$ to be any fixed input. 
%We employ two fundamental theorems in asymptotic theory of statistics: Central Limit Theorem, and Slutsky's Theorem.

\subsubsection{Empirical CDF from original data}\label{appendix-subsubsec:original}
Our empirical CDF for $F_{i}(x)$ is written by: 
\begin{equation}\label{appendix-regional-cdf}
\hat{F}_{i}(x) = \frac{1}{n_{i}}\sum_{j=1}^{n_{i}}I(y_{ij}\leq x).
\end{equation}
For a large number of samples in each region (with sufficiently large $n_{i}$ for every $i$), we apply the Central Limit Theorem under Assumptions~\ref{assum-1} and~\ref{assum-2}. Then,
\begin{equation}\label{appendix-clt-region}
    \hat{F}_{i}(x)\sim^{\text{approx}}\mathcal{N}\left(F_{i}(x), \frac{F_{i}(x)(1-F_{i}(x))}{n_{i}}\right),
\end{equation}
for $i=1,\cdots,k$. Under Assumption~\ref{assum-1} and by the fact that $I(y_{ij}\leq x)$ independently follows a Bernoulli distribution with success probability $F_{i}$ for all $i,j$,  variance of $\hat{F}(x)$ follows: 
\begin{equation}\label{appendix-asym-var}
    \begin{aligned}
    \text{Var}(\hat{F}(x)) & = \text{Var}\left(\sum_{i=1}^{k}\frac{n_{i}}{n}\hat{F}_{i}(x)\right)\\
    & = \sum_{i=1}^{k}\text{Var}\left(\frac{n_{i}}{n}\hat{F}_{i}(x)\right)\\ %\quad\text{(under Assumption~\ref{assum-1}})\\
    & = \sum_{i=1}^{k}\frac{n_{i}^{2}}{n^{2}}\text{Var}\left(\hat{F}_{i}(x)\right)\\
    & = \sum_{i=1}^{k}\frac{n_{i}^{2}}{n^{2}}\frac{F_{i}(x)(1-F_{i}(x))}{n_{i}}\\ %\quad(\forall i,\ I(y_{ij}\leq x)\sim^{\text{indep.}}\text{Bernoulli}(F_{i}(x))\text{ for all }j)\\
    & = \frac{1}{n^{2}}\sum_{i=1}^{k}n_{i}F_{i}(x)(1-F_{i}(x)). 
    \end{aligned}
\end{equation}

Based on the asymptotic normality in~\eqref{appendix-clt-region} and independence based on Assumption~\ref{assum-1}, 
%the asymptotic distribution of a linear combination of~\eqref{appendix-regional-cdf} weighted by the sample proportion,
we have:
\begin{equation}\label{appendix-asym}
    \hat{F}(x)\sim^{\text{approx}}\mathcal{N}\left(\sum_{i=1}^{k}\frac{n_{i}}{n}F_{i}(x), \frac{1}{n^{2}}\sum_{i=1}^{k}n_{i}F_{i}(x)(1-F_{i}(x))\right),
\end{equation}
where its asymptotic variance is obtained by~\eqref{appendix-asym-var}.
Note that $\hat{F}_{i}(x)$ converges in probability to $F_{i}(x)$, i.e. $\hat{F}_{i}(x)\xrightarrow{\mathbb{P}}F_{i}(x)$, under Assumptions~\ref{assum-1} and~\ref{assum-2} for $i=1,\cdots,k$. %From~\eqref{appendix-asym} 
By  Slutsky's theorem,  we have the following expressions:
\begin{equation*}
    \hat{F}(x)\sim^{\text{approx}}\mathcal{N}\left(\sum_{i=1}^{k}\frac{n_{i}}{n}F_{i}(x), \frac{1}{n^{2}}\sum_{i=1}^{k}n_{i}\hat{F}_{i}(x)(1-\hat{F}_{i}(x))\right).
\end{equation*}
Therefore, the 95\% confidence interval for the simple empirical CDF follows: 
\begin{equation*}
    \hat{F}(x)\pm 1.96\frac{\sqrt{\sum_{i=1}^{k}n_{i}\hat{F}_{i}(x)(1-\hat{F}_{i}(x))}}{n}.
\end{equation*}

\subsubsection{Re-weighted empirical CDF from original data}\label{appendix-subsubsec:reweight}

Under Assumption~\ref{assum-1} and by the fact that $I(y_{ij}\leq x)$ independently follows a Bernoulli distribution with success probability $F_{i}$ for all $i,j$, the variance of~(\ref{eq-cdf-reweight}) follows: 
\begin{equation}\label{appendix-asym-var-un}
    \begin{aligned}
        \text{Var}\left(\hat{F}_{u}(x)\right) & = \text{Var}\left(\sum_{i=1}^{k}\frac{N_{i}}{N}\hat{F}_{i}(x)\right)\\
        & =\sum_{i=1}^{k}\text{Var}\left(\frac{N_{i}}{N}\hat{F}_{i}(x)\right)\\ %\quad(\text{under Assumption~\ref{assum-1})}\\
        & = \sum_{i=1}^{k}\frac{N_{i}^{2}}{N^{2}}\text{Var}\left(\hat{F}_{i}(x)\right)\\
        & = \sum_{i=1}^{k}\frac{N_{i}^{2}}{N^{2}}\frac{F_{i}(x)(1-F_{i}(x))}{n_{i}}\\ %\quad(\forall i,\ I(y_{ij}\leq x)\sim^{\text{indep.}}\text{Bernoulli}(F_{i}(x))\text{ for all }j)
    \end{aligned}
\end{equation}

%Again, 
\noindent
Since the empirical CDFs from~\eqref{appendix-regional-cdf} are asymptotically normal, as shown in~\eqref{appendix-clt-region}, with independence between regions (Assumption~\ref{assum-1}), the re-weighted sum of these CDFs are asymptotically normal, i.e.

\begin{equation}\label{appendix-asym-un}
    \hat{F}_{u}(x)\sim^{\text{approx}}\mathcal{N}\left(\sum_{i=1}^{k}\frac{N_{i}}{N}F_{i}(x), \frac{1}{N^{2}}\sum_{i=1}^{k}\frac{N_{i}^{2}}{n_{i}}F_{i}(x)(1-F_{i}(x))\right),
\end{equation}
where the asymptotic variance is obtained by~\eqref{appendix-asym-var-un}.

Since $\hat{F}_{i}(x)$ converges in probability to $F_{i}(x)$, i.e. $\hat{F}_{i}(x)\xrightarrow{\mathbb{P}}F_{i}(x)$, under Assumptions~\ref{assum-1} and~\ref{assum-2} for $i=1,\cdots,k$, we employ Slutsky's theorem to bridge~\eqref{appendix-asym-un} to obtain: 
%the conclusion that
\begin{equation*}
    \hat{F}_{u}(x)\sim^{\text{approx}}\mathcal{N}\left(\sum_{i=1}^{k}\frac{N_{i}}{N}F_{i}(x), \frac{1}{N^{2}}\sum_{i=1}^{k}\frac{N_{i}^{2}}{n_{i}}\hat{F}_{i}(x)(1-\hat{F}_{i}(x))\right).
\end{equation*}
Then the 95\% confidence interval for cdf $F(x)$  is obtained by:
$$\hat{F}_{u}(x)\pm 1.96\frac{\sqrt{\sum_{i=1}^{k}\frac{N_{i}^{2}}{n_{i}}\hat{F}_{i}(x)(1-\hat{F}_{i}(x))}}{N}$$

\subsubsection{Simple empirical CDF from re-sampled data}\label{appendix-subsubsec:resample}

Following the notation in Section~\ref{subsubsec:re-sampling} and the derivation in Appendix~\ref{appendix-subsubsec:original}, we obtain the 95\% confidence interval for $\hat{F}^{*}$ as:

\begin{equation*}
    \hat{F}^{*}(x)\pm 1.96\frac{\sqrt{\sum_{i=1}^{k}n_{i}^{*}\hat{F}_{i}^{*}(x)(1-\hat{F}^{*}_{i}(x))}}{n^{*}}.
\end{equation*}

\subsection{Proof of Lemma~\ref{lemma::sufficiency}}\label{appendix-suff-proof}

\begin{proof}
    %To prove the equality, it suffices to express both likelihood functions, $l(\boldsymbol \beta, \sigma^{2})$ and $\bar{l}(\boldsymbol \beta, \sigma^{2})$ in a relatable way. 
    From models~\eqref{eq-model-indi} and~\eqref{eq-model-agg}, we have $y_{ij}\sim^{\text{indep.}}\mathcal{N}(\mathbf x_{i}^{T}\boldsymbol \beta, \sigma^{2})$ and $\bar{y}_{i}\sim^{\text{indep.}}\mathcal{N}(\mathbf x_{i}^{T}\boldsymbol \beta, \sigma^{2}/n_{i})$ for $j=1,\cdots,n_{i}$ and $i=1,\cdots,k$. First, note that: 
    \begin{equation}\label{appendix-eq-L-bar}
            \bar{l}(\boldsymbol \beta, \sigma^{2}) = -\frac{k}{2}\log(2\pi\sigma^{2})+\frac{1}{2}\sum_{i}^{k}\log n_{i}-\frac{1}{2\sigma^{2}}\sum_{i}^{k}n_{i}(\bar{y}_{i}-\mathbf x_{i}^{T}\boldsymbol \beta)^{2}.
    \end{equation}
    On the other hand, we also have: 
    \begin{equation}\label{appendix-eq-L}
     \begin{split}
          & l(\boldsymbol \beta, \sigma^{2}) = -\frac{n}{2}\log(2\pi\sigma^{2})-\frac{1}{2\sigma^{2}}\sum_{i=1}^{k}\sum_{j=1}^{n_{i}}(y_{ij}-\mathbf x_{i}^{T}\boldsymbol \beta)^{2} \\
        & = -\frac{n}{2}\log(2\pi\sigma^{2})-\frac{1}{2\sigma^{2}}\sum_{i=1}^{k}\sum_{j=1}^{n_{i}}(y_{ij}-\bar{y}_{i}+\bar{y}_{i}-\mathbf x_{i}^{T}\boldsymbol \beta)^{2}\\
        & =  -\frac{n}{2}\log(2\pi\sigma^{2})-\frac{1}{2\sigma^{2}}\sum_{i=1}^{k}\sum_{j=1}^{n_{i}}(y_{ij}-\bar{y}_{i})^{2}-\frac{1}{2\sigma^{2}}\sum_{i=1}^{k}\sum_{j=1}^{n_{i}}(\bar{y}_{i}-\mathbf{x}_{i}^{T}\boldsymbol{\beta})^{2}-\frac{1}{\sigma^{2}}\sum_{i=1}^{k}\sum_{j=1}^{n_{i}}(y_{ij}-\bar{y}_{i})(\bar{y}_{i}-\mathbf x_{i}^{T}\boldsymbol{\beta})\\
        & = -\frac{n}{2}\log(2\pi\sigma^{2})-\frac{1}{2\sigma^{2}}\sum_{i=1}^{k}\sum_{j=1}^{n_{i}}(y_{ij}-\bar{y}_{i})^{2}-\frac{1}{2\sigma^{2}}\sum_{i=1}^{k}n_{i}(\bar{y}_{i}-\mathbf{x}_{i}^{T}\boldsymbol{\beta})^{2}-\frac{1}{\sigma^{2}}\sum_{i=1}^{k}\sum_{j=1}^{n_{i}}(y_{ij}-\bar{y}_{i})(\bar{y}_{i}-\mathbf x_{i}^{T}\boldsymbol{\beta})\\
        & = -\frac{n}{2}\log(2\pi\sigma^{2})-\frac{1}{2\sigma^{2}}\sum_{i=1}^{k}\sum_{j=1}^{n_{i}}(y_{ij}-\bar{y}_{i})^{2}-\frac{1}{2\sigma^{2}}\sum_{i=1}^{k}n_{i}(\bar{y}_{i}-\mathbf{x}_{i}^{T}\boldsymbol{\beta})^{2},
     \end{split}
    \end{equation}
    as $\sum_{j=1}^{n_{i}}(y_{ij}-\bar{y}_{i})(\bar{y}_{i}-\mathbf x_{i}^{T}\boldsymbol{\beta}) = (\bar{y}_{i}-\mathbf x_{i}^{T}\boldsymbol{\beta})\sum_{j=1}^{n_{i}}(y_{ij}-\bar{y}_{i}) =(\bar{y}_{i}-\mathbf x_{i}^{T}\boldsymbol{\beta})(n_{i}\bar{y}_{i}-n_{i}\bar{y}_{i}) = 0 $ for each $i=1,\cdots,k$. 
    Comparing with Equation (\ref{appendix-eq-L-bar}), the results follow. 
    %The results in Lemma then follow. 
%Then, the expression of $c_{\sigma^{2}}$ from~\eqref{eq-lemma::suff} in Lemma~\ref{lemma::sufficiency} naturally follows from the $l(\boldsymbol{\beta},\sigma^{2})-\bar{l}(\boldsymbol{\beta},\sigma^{2})$, the calculation of which can be easily done from~\eqref{appendix-eq-L-bar} and~\eqref{appendix-eq-L}.
    
\end{proof}

\subsection{Proof of Lemma~\ref{lemma::efficiency}}

%We assume that the value of predictors is given. 
We first introduce the following lemma to derive estimators of $\sigma^{2}$ in the model~\eqref{eq-model-indi} and~\eqref{eq-model-agg}, and their variability. Here, for the generality, let $\mathbf x_{i} = (1, x_{i1}, x_{i2}, \cdots, x_{ip})^{T}\in\mathbb{R}^{(p+1)}$ represent the predictor vector of $i$-th region with $p$ different predictors for $i=1,\cdots,k\geq 2$, and $p>2$. We assume that $n>>k>p+1$, indicating that the total number of sample across the collected regions is significantly greater than the number of regions, and the dimension of feature space does not exceed the number of regions.

\begin{lemma}\label{appendix-quadratic}
Let $\mathbf y$ be a $n$-dimensional random vector with $y\sim\mathcal{N}(\boldsymbol \mu, \mathbf I_n)$, where $\boldsymbol \mu \in\mathbb{R}^{n}$ and $\mathbf I_n$ denotes $n\times n$ identity matrix. If $\mathbf M \in\mathbb{R}^{n\times n}$ is an orthogonal projection matrix, then 
$$\mathbf{y}^{T}\mathbf{M}\mathbf y\sim \chi^{2}(r(\mathbf{M}), \boldsymbol \mu^{T}\mathbf M \boldsymbol \mu /2),$$
where $r(\mathbf A)$ indicates the rank of a given square matrix $\mathbf A$ and $\chi^{2}(d,\gamma)$ refers to the noncentral $\chi^{2}$ distribution with degree of freedom $d$ and noncentrality parameter $\gamma$. A noncentral chi-squared distribution $\chi^{2}(d, \gamma)$ is generated by a sum of squared independent Gaussian random variables $z_{1},\cdots,z_{d}\sim\mathcal{N}(\mu, 1)$, i.e. $\sum_{i=1}^{d}z_{i}^{2}$. Here, the noncentrality parameter $\gamma$ is defined by $\gamma = \sum_{i=1}^{d}\mu_{i}^{2}/2$.
\end{lemma}

\begin{proof}
    Let $r(\mathbf M)=r$ and let $\mathbf b_{1},\cdots,\mathbf b_{r}\in\mathbb{R}^{n}$ be an orthonormal basis for the column space of $\mathbf M$, say $\mathcal{C}(\mathbf M)$. Let $\mathbf B = [\mathbf b_{1},\cdots, \mathbf b_{r}]\in\mathbf{R}^{n\times r}$ so that $\mathbf M =\mathbf B \mathbf B^{T}$. We have $\mathbf y^{T}\mathbf M \mathbf y = \mathbf y \mathbf B \mathbf B^{T} \mathbf y = (\mathbf B^{T}\mathbf y)^{T}(\mathbf B^{T}\mathbf y)$, where $\mathbf B^{T}\mathbf y\sim \mathcal{N}(\mathbf B^{T}\boldsymbol \mu, \mathbf{B}^{T}\mathbf{B})$. Since the columns of $\mathbf{B}$ are orthonormal, $\mathbf{B}^{T}\mathbf B = \mathbf I$. By definition of noncentral $\chi^{2}$ distribution, $(\mathbf B^{T}\mathbf y)^{T}(\mathbf{B}^{T}\mathbf y)\sim \chi^{2}(r, \boldsymbol \mu \mathbf B \mathbf B^{T}\boldsymbol \mu/2$) where $\boldsymbol \mu^{T}\mathbf B \mathbf B^{T}\boldsymbol \mu = \boldsymbol \mu^{T}\mathbf M \boldsymbol \mu$.
\end{proof}

\subsubsection{Unbiasedness of $\hat{\sigma}^{2}$}

%Define a response vector $\mathbf y = (y_{11},\cdots,y_{1n_{1}}, \cdots,y_{k1},\cdots,y_{kn_{k}})^{T}\in\mathbb{R}^{n}$. Let $\mathbf{V} = [\mathbf x_{1} \mathbf{1}_{n_{1}}^{T}  ,\mathbf x_{2} \mathbf{1}_{n_{2}}^{T},\cdots, \mathbf x_{k} \mathbf{1}_{n_{k}}^{T}]^{T}\in\mathbb{R}^{n\times (p+1)}$ where $\mathbf{1}_{d}$ denotes $d$-dimensional column vector with all elements being 1. Define $\boldsymbol \epsilon = (\epsilon_{11},\cdots,\epsilon_{1n_{1}},\cdots,\epsilon_{k1},\cdots,\epsilon_{kn_{k}})^{T}\in\mathbb{R}^{n}$. Then, the linear model in~\eqref{eq-lm-indi} can be expressed as
%\begin{equation}\label{appendix-eq-lm-indi}
%    \mathbf y = \mathbf V \boldsymbol \beta +\boldsymbol \epsilon,
%\end{equation}

%where $\epsilon\sim\mathcal{N}(\mathbf 0, \sigma^{2}\mathbf{I}_{n})$. We may assume that $\mathbf V$ is of full column rank. 
%Consider the model in~\eqref{eq-model-indi-matrix}. We may assume that 
W.l.o.g., we only consider $\mathbf V$ has full column rank, given that distinctive features among regions. Define $\mathbf J = \mathbf V(\mathbf V^{T}\mathbf V)^{-1}\mathbf V^{T}$ as the orthogonal projection matrix onto $\mathcal{C}(\mathbf V)$ with $r(\mathbf J) = p+1$. It follows that $(\mathbf{I}_{n}-\mathbf V)$ is the orthogonal projection matrix onto $\mathcal{C}(\mathbf{V})^{\perp}$ with $r(\mathbf I_{n}-\mathbf J) = n-p-1$. By Lemma~\ref{appendix-quadratic}, we have 
\begin{equation*}
    \mathbf y^{T}(\mathbf{I}_{n}-\mathbf{J})\mathbf y /\sigma^{2}\sim \chi^{2}(n-p-1),
\end{equation*}
where $\chi^{2}(n-p-1)$ is a chi-sqaured distribution with degree of freedom $n-p-1$ and non-central parameter being zero, 
because $(\mathbf I_{n}-\mathbf J)\mathbf{V} = 0$. Then, one can construct an unbiased estimator: 
\begin{equation}\label{appendix-un-indi}
    \hat{\sigma}^{2} = \frac{\mathbf y^{T} (\mathbf I_{n}-\mathbf J)\mathbf y}{n-p-1}, 
\end{equation}
for $\sigma^{2}$ as $\mathbb{E}\left[\mathbf y^{T}(\mathbf{I}_{n}-\mathbf{J})\mathbf y /\sigma^{2}\right] = n-p-1$.

\subsubsection{Unbiasedness of $\hat{\sigma}^{2}_{\text{agg}}$}
%Define $\mathbf X = [\mathbf x_{1}, \cdots,\mathbf x_{k}]^{T}\in\mathbb{R}^{k\times (p+1)}$ and assume that $\mathbf X$ is of full column rank Then, we can write an aggregated version of model~\eqref{appendix-eq-lm-indi} as
%\begin{equation}\label{appendix-eq-lm-agg}
%\bar{\mathbf y} = \mathbf X \boldsymbol \beta +\boldsymbol \epsilon_{\text{agg}},
%\end{equation}
%where $\boldsymbol \epsilon_{\text{agg}}\sim\mathcal{N}(\mathbf{0},\sigma^{2}\mathbf W^{-2})$ and $\mathbf{W}$ is a $k\times k$ diagonal matrix with diagonal entries being $\sqrt{n_{1}},\cdots,\sqrt{n_{k}}$. From the model equation~\eqref{appendix-eq-lm-agg}, 

From the model~\eqref{eq-model-agg-matrix}, we have
$\mathbf W \mathbf y\sim\mathcal{N}(\mathbf W\mathbf X\boldsymbol \beta, \sigma^{2}\mathbf{I}_{k})$. Assume that $\mathbf X$ is of full column rank. Define $\mathbf H = \mathbf W\mathbf X(\mathbf X\mathbf W^{2}\mathbf X)^{-1}\mathbf X^{T}\mathbf W$ so that $\mathbf H$ is the orthogonal projection matrix onto $\mathcal{C}(\mathbf{WX}$) with $r(\mathbf{WX}) = p+1$. Similarly, $(\mathbf{I}_{k}-\mathbf{H})$ is an orthogonal matrix onto $\mathcal{C}(\mathbf{WX})^{\perp}$ with $r(\mathbf{I}_{k}-\mathbf H) = k-p-1$. By Lemma~\ref{appendix-quadratic}, it follows that
\begin{equation*}
    \mathbf y^{T}\mathbf W(\mathbf I_{k}-\mathbf H)\mathbf W\mathbf y/\sigma^{2}\sim\chi^{2}(k-p-1),
\end{equation*}
%where $\chi^2$ distribution with degree of freedom $k-p-1$ with noncentrality parameter zero 
as $(\mathbf I_{k}-\mathbf H)\mathbf W\mathbf X = 0$. Then, we can use an unbiased estimator 
\begin{equation}\label{appendix-un-agg}
    \hat{\sigma}^{2}_{\text{agg}} = \frac{\mathbf y^{T}\mathbf W(\mathbf I_{k}-\mathbf H)\mathbf W\mathbf y}{k-p-1}
\end{equation}
for estimate $\sigma^{2}$ since $\mathbb{E}\left[\mathbf y^{T}\mathbf W(\mathbf I_{k}-\mathbf H)\mathbf W\mathbf y/\sigma^{2}\right] = k-p-1$.

\subsubsection{Quantification of efficiency}
Note that if a random variable $W$ follows a $\chi^2$ distribution with degree of freedom $d$, then $\text{Var}(W) = 2d$. As such, we can conclude that even though the estimators in~\eqref{appendix-un-indi} and~\eqref{appendix-un-agg} are both unbiased for $\sigma^{2}$, the estimator in~\eqref{appendix-un-indi} is more efficient than one in~\eqref{appendix-un-agg} because $\text{Var}(
\hat{\sigma}^{2}) = 2\sigma^{4}/(n-p-1) < 2\sigma^{4}/(k-p-1) = \text{Var}\left(\hat{\sigma}^{2}_{\text{agg}}\right)$, and this difference is noticeable since the number  of regions $k$ is much smaller than the total number of samples $n$.

\begin{figure}[t!]
  \centering
  \includegraphics[width=\linewidth]{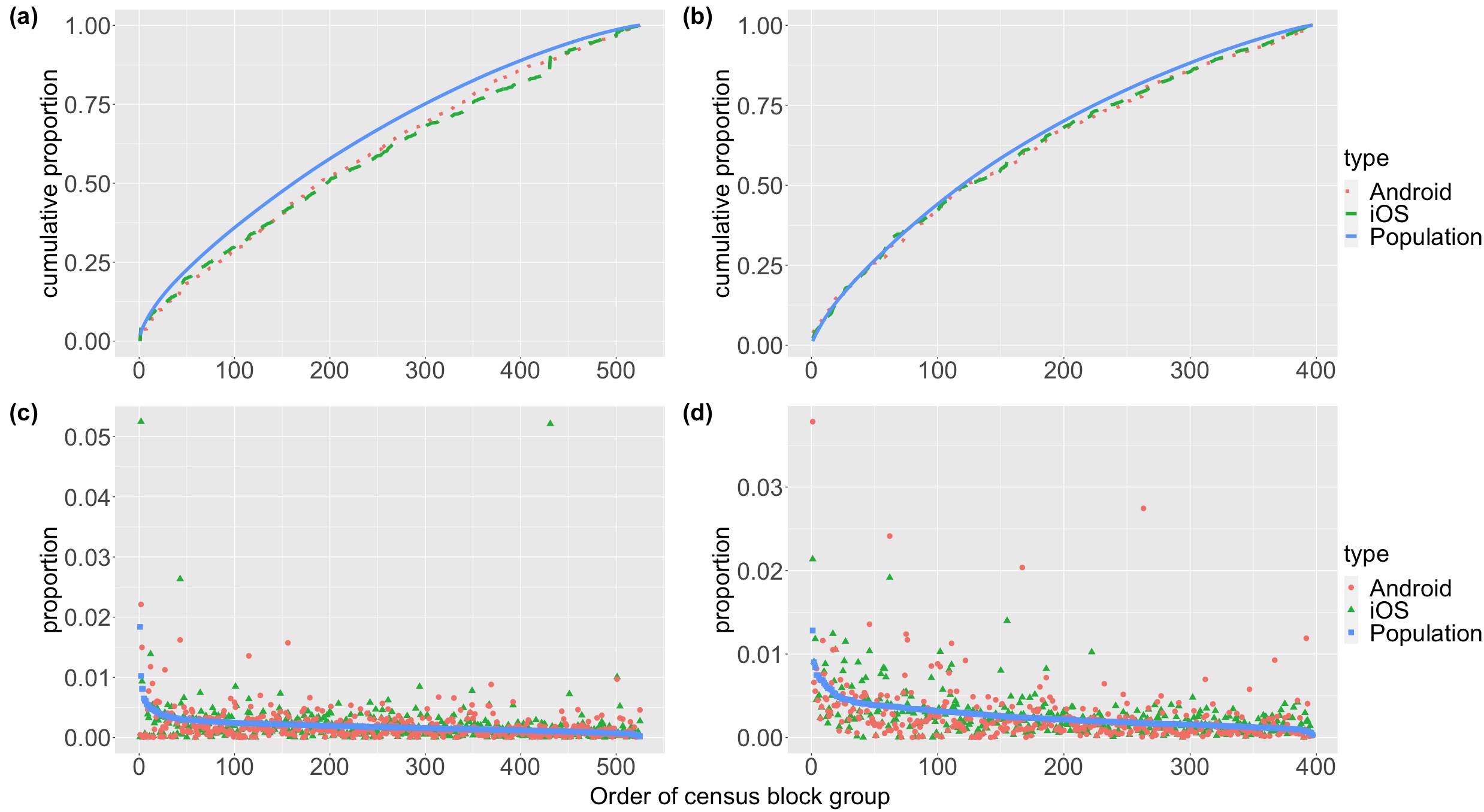}
  \caption{\label{appendix-fig-pop-sam}Comparison of the distributions of population and sample sizes. Census block groups are ordered from the one with the largest population to the one with the smallest population.  Cumulative distribution functions are of population/sample sizes from iOS and Android devices at different census block groups in City C (part (a)) and City D (part (b)). The 
   probability mass functions are of population/sample sizes from iOS and Android devices in City C (part (c)) and City D (part (d)). The number on the x-axis indicates the rank of the population size in each census block group; the census block group with the largest population has rank 1. } 
\end{figure}

\vspace*{0.1in}
\section{Other cities and devices types}\label{appendix-sec:othercities}

\subsection{Regional sampling bias detection and correction}\label{appendix-subsec:bias-detection}

Here we provide additional numerical results of the regional sampling bias in two other cities, referred to as City C and City D. We also examine whether the regional sampling bias affects the estimation of internet speed and its association with regional demographic profiles, as discussed in Section~\ref{sec:bias-detection}. 
%We observe that sample sizes do not align with population of corresponding census block groups, and the difference seems to be more noticeable for City C,  
We compare the cumulative distribution of sample sizes with population in Fig.~\ref{appendix-fig-pop-sam} for City C and D, and the difference seems to be more noticeable for City C. 
%and this deviance is more severe for Android than iOS devices. 

\begin{figure}[t!]
  \centering
  \includegraphics[width=\linewidth]{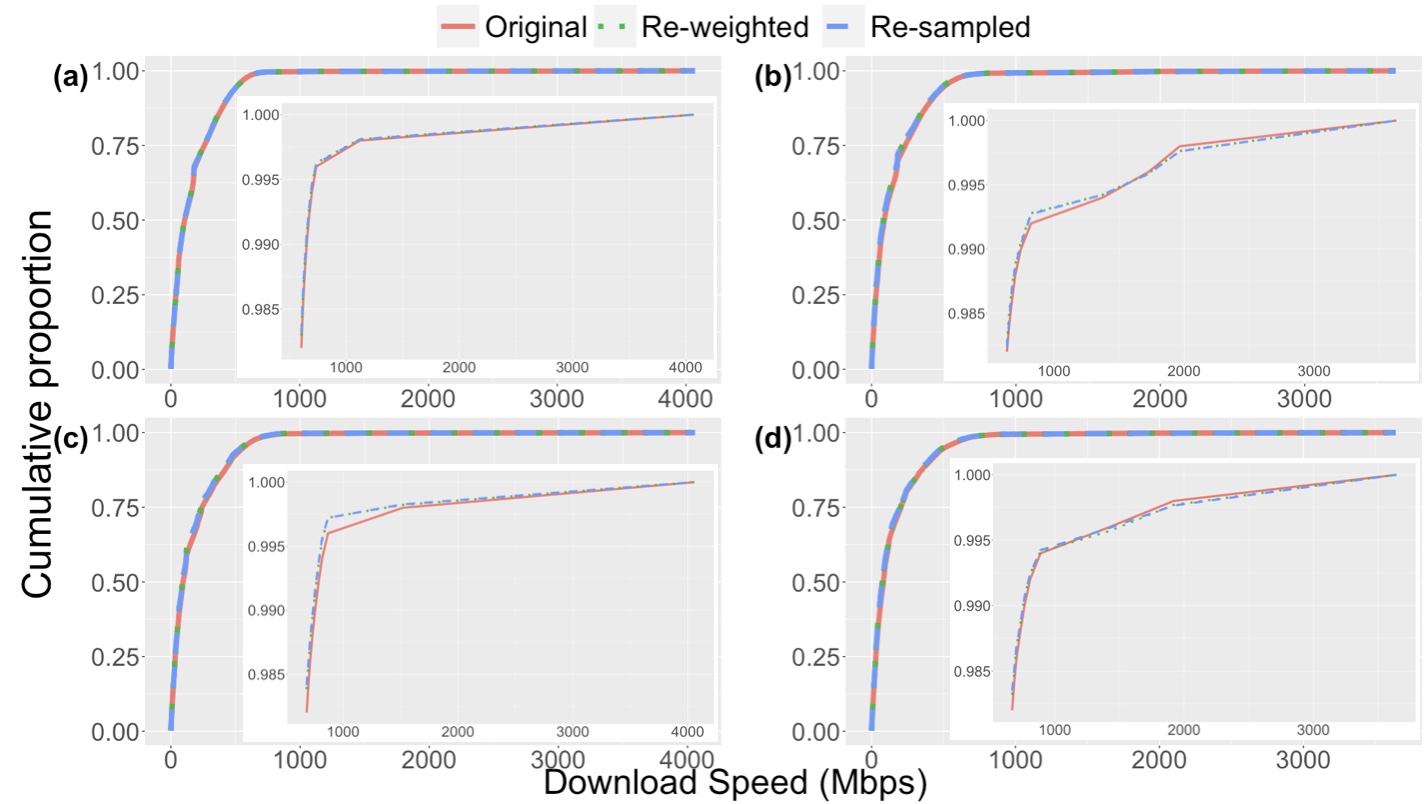}
  \caption{\label{appendix-fig-cdf-compare}
  Comparison of empirical CDFs of the internet download speed from the original samples with re-weighted empirical CDF from original samples, and empirical CDF from re-sampled data. (a) iOS devices in City C; (b) Android devices in City C; (c) iOS devices in City D; and (d) Android devices in City D. The insets show zoom-in plots of the estimated CDFs of the internet download speed above the 98th percentile from the original samples.}
\end{figure}

\begin{table*}[ht!]
\centering
  \caption{The result from $\chi^{2}$ homogeneity test for City C and D.}
  \label{appendix-tab:chi-test}
  \begin{tabular}{ccccc}
    \toprule
    City & \# of regions & Device type & Test statistics $W$ & p-value\\
    \midrule
    \multirow{2}{*}{C} & \multirow{2}{*}{525} &                   iOS & 780,884 & $<10^{-16}$\\
              &       & Android & 128,545 & $<10^{-16}$\\
    \midrule
    \multirow{2}{*}{D}& \multirow{2}{*}{397} &                  iOS & 88,953 &                                 $<10^{-16}$\\
                &    & Android & 142,934 &     $<10^{-16}$\\
    \bottomrule
  \end{tabular}
\end{table*}

This deviance is statistically evaluated by the chi-squared homogeneity test as shown in Table~\ref{appendix-tab:chi-test}. For each device type in either City C or D, we find that the proportion of sample sizes is significantly different from the proportion of the population. We further compare the cumulative distribution functions of internet speed in Fig.~\ref{appendix-fig-cdf-compare} based on three different empirical CDFs introduced in Section~\ref{subsec:cdf}. Neither re-weighing and re-sampling methods  show a notable  difference from the empirical CDF with regional sampling bias, which is consistent with the findings for City A and  B. 
%and this suggests that its impact on the estimation of internet speed distribution is tangential in City C and D.

The demographic disparity between the over- and under-sampled census block groups of City C and City D for iOS devices is visualized in  Fig.~\ref{appendix-fig-boxplot}.
%as discussed in Section~\ref{subsec:boxplot}. 
We find that in both City C and D that the over-sampled regions tend to have higher income, greater age,  a larger proportion of population with a bachelor's degree or higher, a greater percentage of households with internet subscription plans, higher representation of white and Asian residents; and lower representation of black and Hispanic residents. 
%We also compare the over- and under-sampled regions in Fig.~\ref{appendix-fig-boxplot-android-CD} for Android devices. The over-sampled regions in City C are likely to have higher age, and larger population with bachelor's degree or higher, and less Hispanic population whereas the over-sampled regions in City D tend to have less Hispanic population.
The differences by demographic profiles between over-sampled and under-sampled regions across all four cities generally agree with each other.

\begin{figure}[t!]
  \centering
  \includegraphics[width=.9\linewidth]{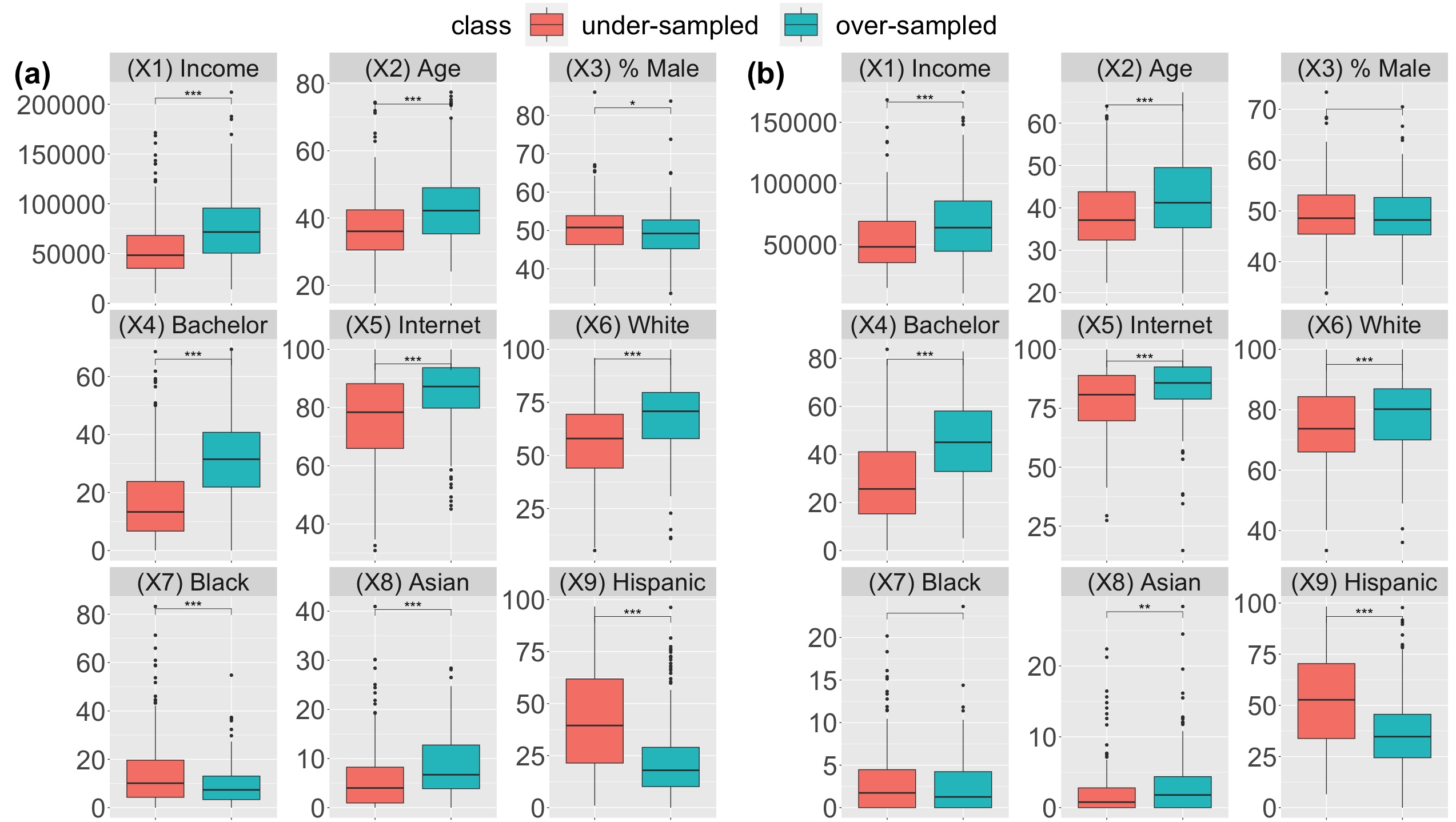}
  \caption{\label{appendix-fig-boxplot}
  Comparison of  distributions of demographic variables between over-/under-sampled census block groups. Significance is based on the p-values from the two-sample t-test: *** (p-value$<$0.001); ** (p-value$<$0.01); * (p-value$<$0.05); and empty mark for non-significant cases. (a) iOS devices from City C; (b) iOS devices from City D. }
\end{figure}

\begin{figure}[t!]
  \centering
  \includegraphics[width=.9\linewidth]{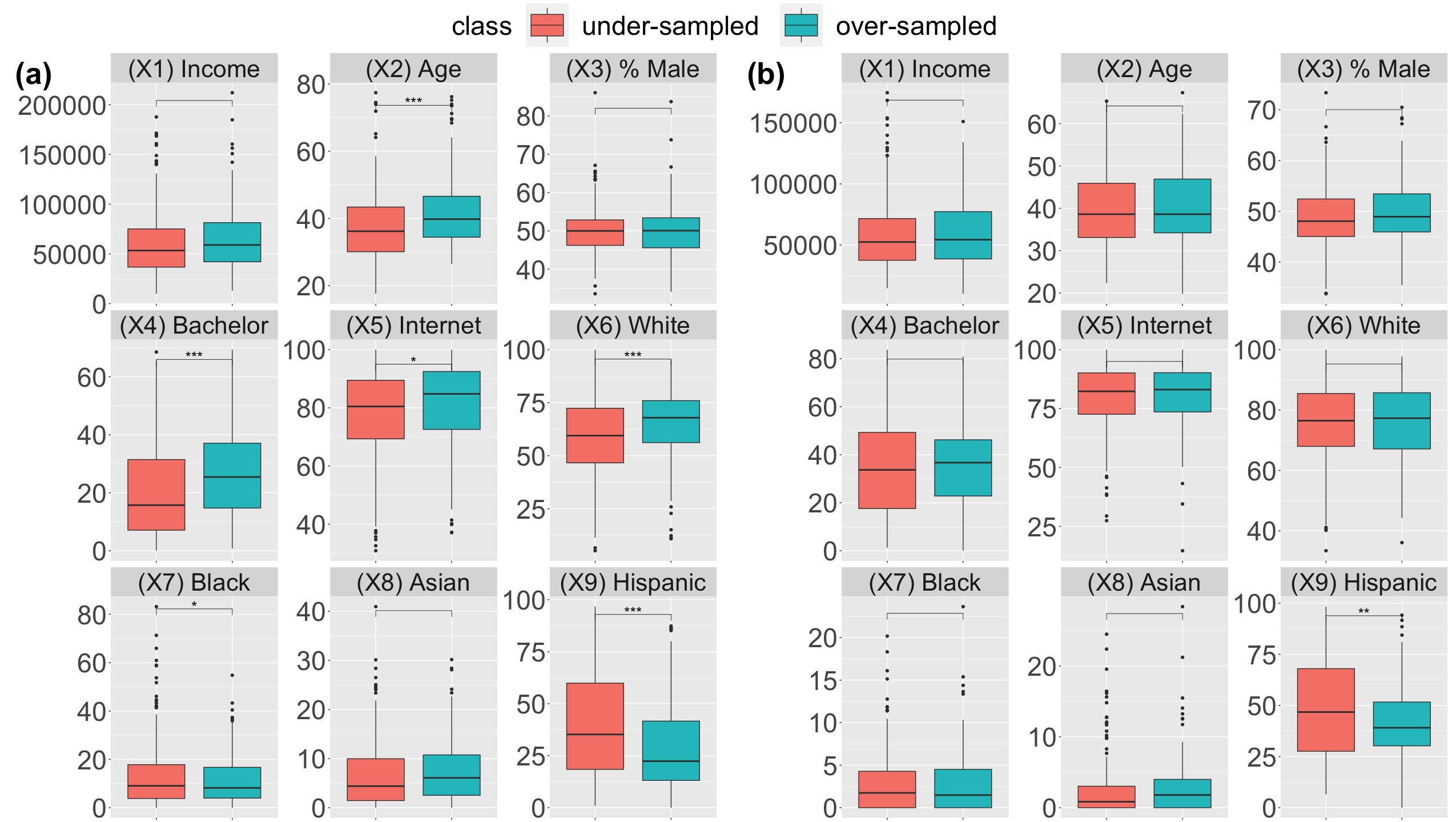}
  \caption{\label{appendix-fig-boxplot-android-CD}
  Comparison of  distributions of demographic variables between over-/under-sampled census block groups. Significance is based on the p-values from the two-sample t-test: *** (p-value$<$0.001); ** (p-value$<$0.01); * (p-value$<$0.05); and empty mark for non-significant cases. (a) Android devices from City C; (b) Android devices from City D. }
\end{figure}

%We complement the analysis in Fig.~\ref{fig-boxplot} of Section~\ref{subsec:boxplot} by adding 
%The comparison of Android devices for the four cities are given in Fig.~\ref{appendix-fig-boxplot-android} and Fig. \ref{appendix-fig-boxplot-android-CD}. 
%to compare over- and under-sampled census block groups for Android devices in City A and B. 
The comparison of Android devices for City C and D is given in Fig.~\ref{appendix-fig-boxplot-android-CD}. 
We find that the over-sampled census block groups tend to have similar demographic characteristics to those of the iOS devices for each city. The difference of demographic profiles between over-sampled regions and under-sampled regions for the Android devices tends to be smaller than those of the iOS devices. 
%in City A and B that the over-sampled census block groups tend to have higher age, larger population with bachelor's or degree or higher, which is in line with the results in Fig.~\ref{fig-boxplot}. 
%The census block groups in City A with more white population and less black residents are likely to have more speed test samples.

\subsection{Correlating internet quality with demographic variables}\label{appendix-subsec:lm}

\begin{figure}[hbt!]
  \centering
  \includegraphics[width=\linewidth]{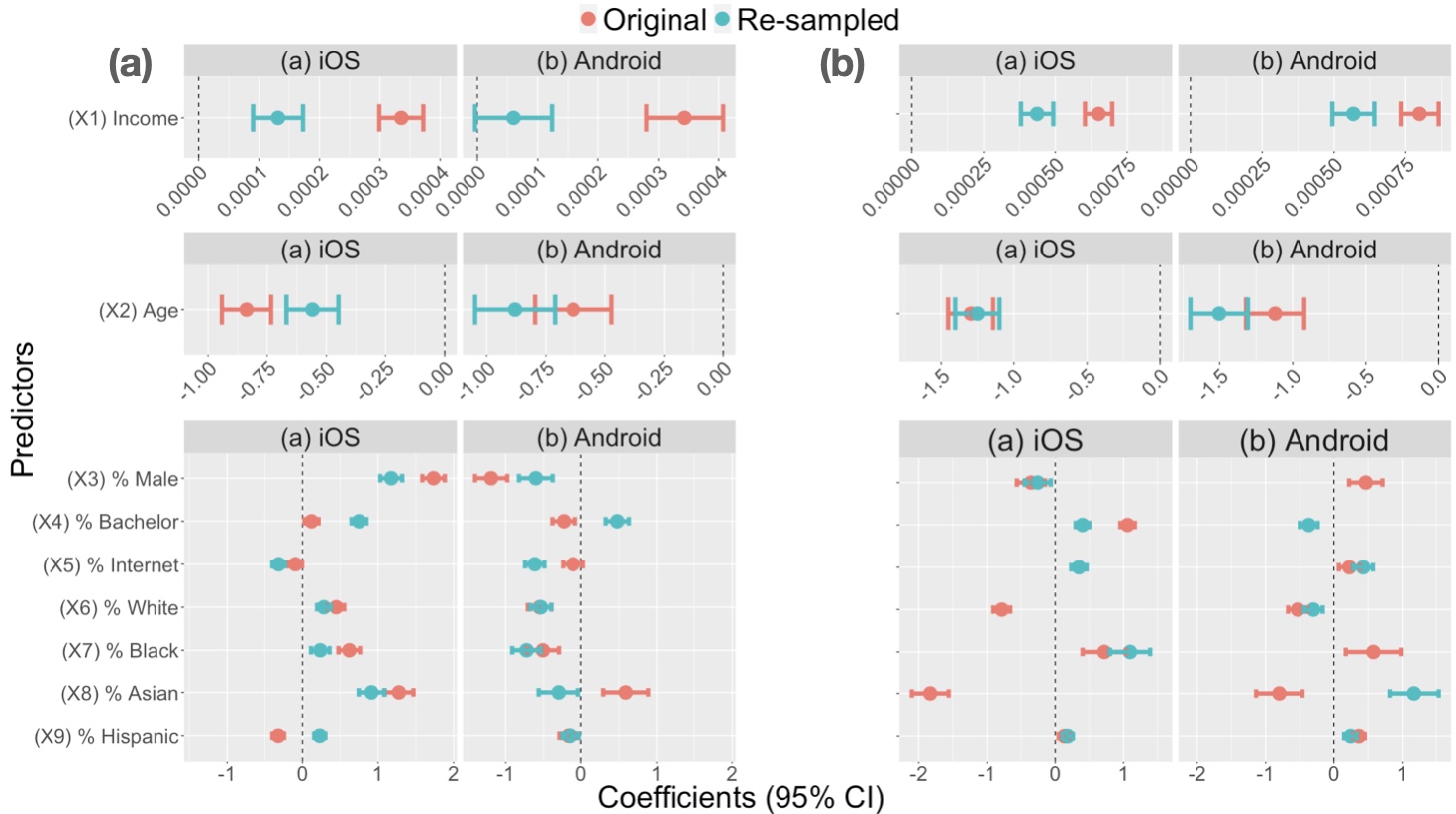}
  \caption{\label{appendix-fig-coef-compare}
  Comparison of regression coefficient estimates (dots) and 95\% confidence intervals (bars) from original data and re-sampled data. Multiple linear regression is conducted with backward variable selection by AIC for both iOS and Android from (a) City C and (b) City D. %95\% confidence interval (CI) is shown with bars. 
  The variables dropped from model selection are not shown in the plot.}
\end{figure}

\begin{figure}[hbt!]
  \centering
  \includegraphics[width=.8\linewidth]{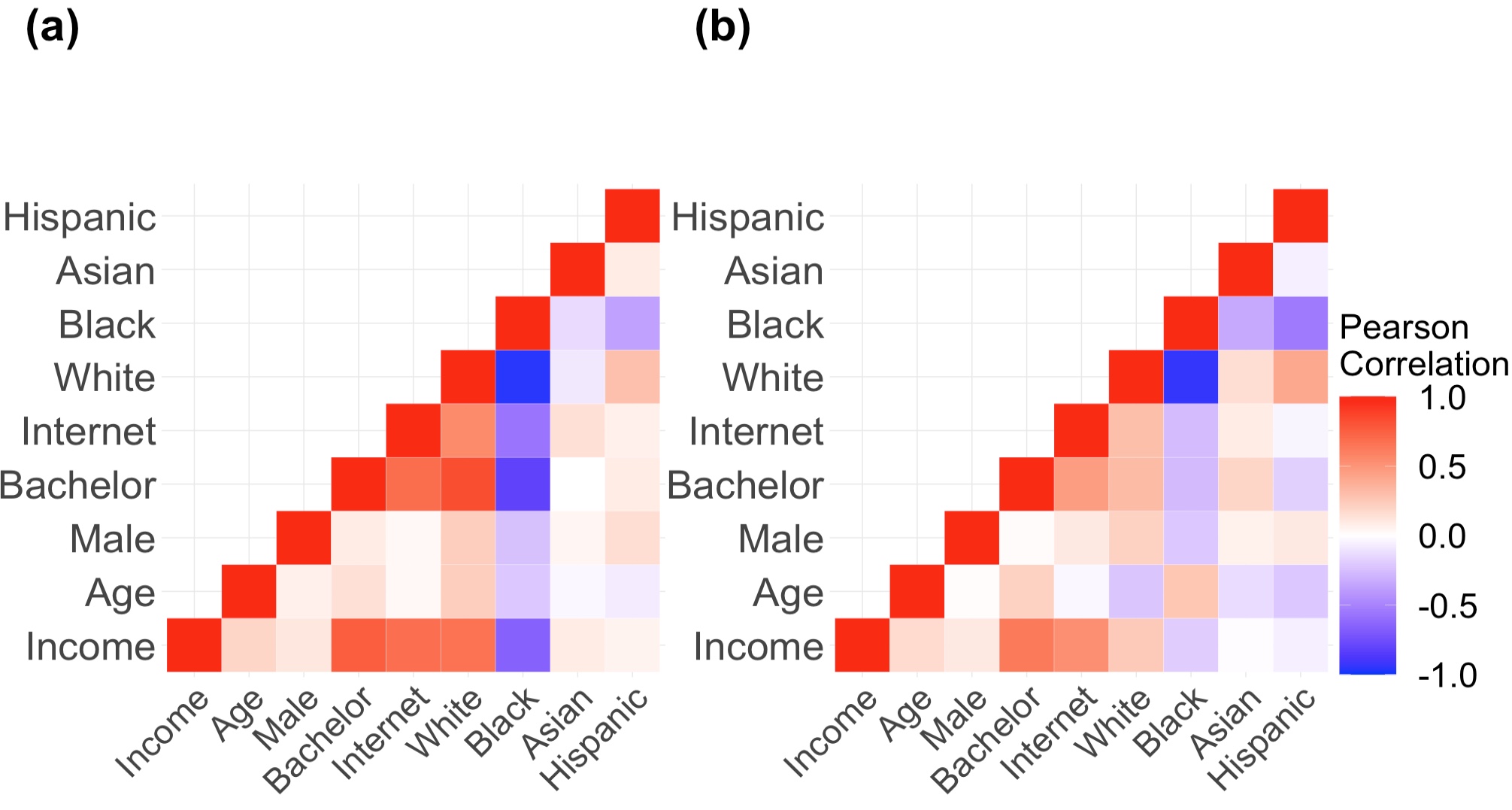}
  \caption{\label{appendix-fig-corrplot_android}
   Heat-map of pair-wise correlation coefficients between demographic covariates from re-sampled data for 
   %from original data. 
   Android devices in
   (a) City A and (b) City B. 
   %The case of Android devices is provided in Appendix
   }
\end{figure}

\begin{figure}[hbt!]
  \centering
  \includegraphics[width=.8\linewidth]{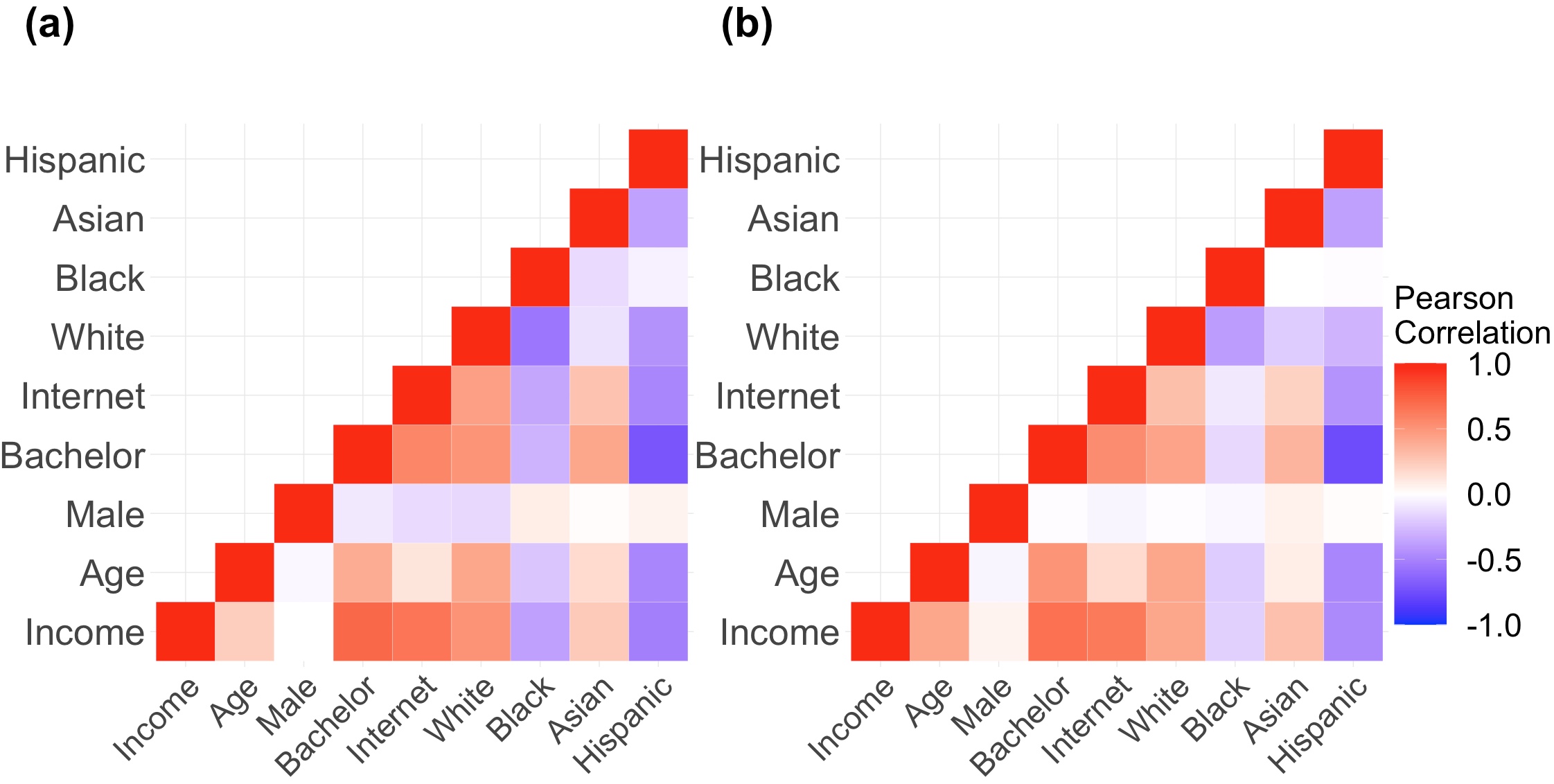}
  \caption{\label{appendix-fig-corrplot}
   Heat-map of pair-wise correlation coefficients between demographic covariates from re-sampled data for 
   %from original data. 
   iOS devices in
   (a) City C and (b) City D. 
   %The case of Android devices is provided in Appendix
   }
\end{figure}

\begin{figure}[hbt!]
  \centering
  \includegraphics[width=.8\linewidth]{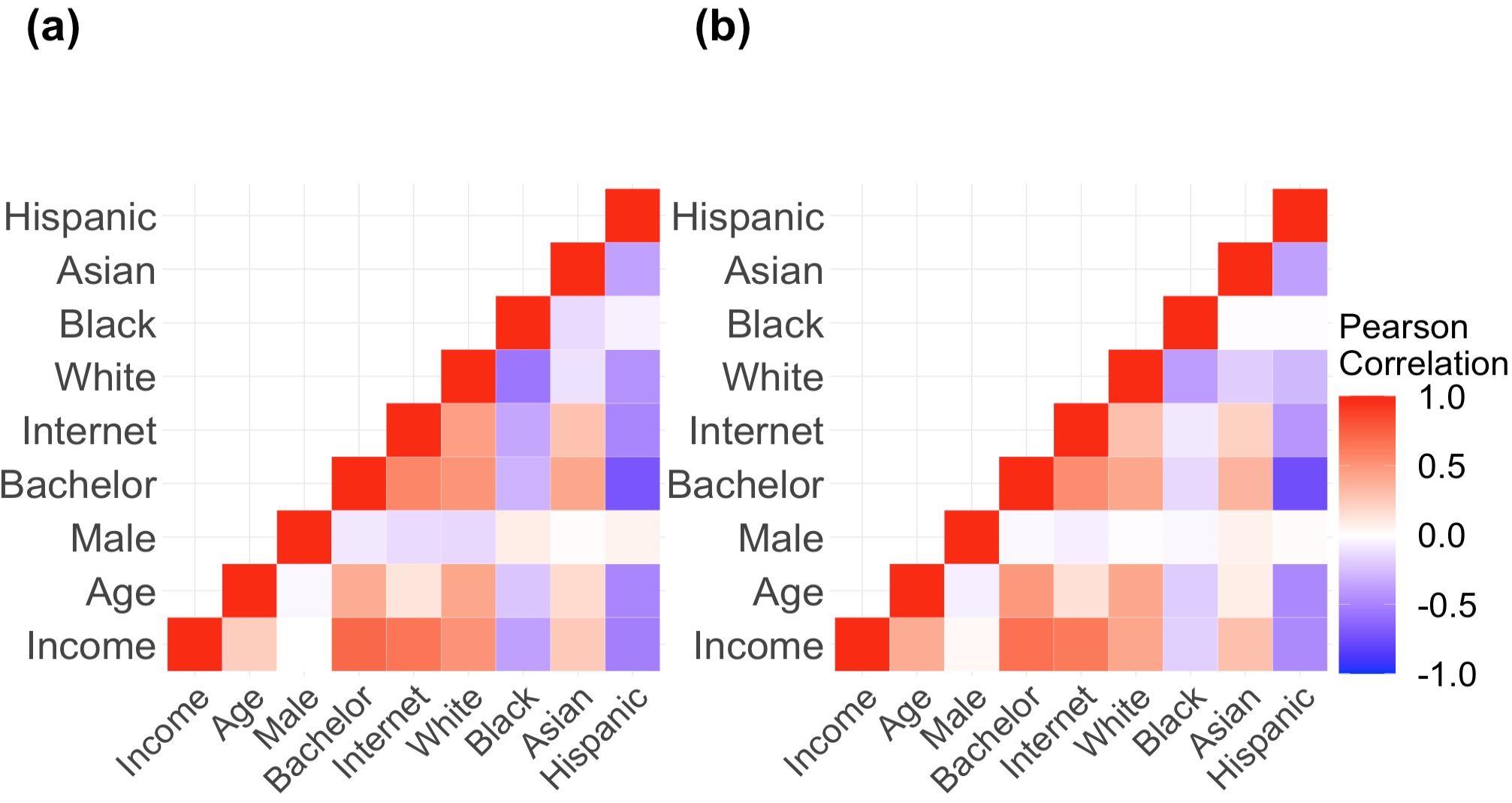}
  \caption{\label{appendix-fig-corrplot_android_CD}
   Heat-map of pair-wise correlation coefficients between demographic covariates from re-sampled data for 
   %from original data. 
   Android devices in
   (a) City C and (b) City D. 
   %The case of Android devices is provided in Appendix
   }
\end{figure}

We applied the regression models~\eqref{eq-model-indi} and~\eqref{eq-model-agg} in Section~\ref{sec:lm} for both iOS and Android devices in City C and D with backward variable selection based on AIC. The estimated coefficients from the selected variables and $95\%$ confidence intervals of the estimation are shown in Fig.~\ref{appendix-fig-coef-compare}. First,  income is positively correlated with internet speed for both cities and devices types, whereas the coefficients of income from the re-sampled data have a smaller impact on the internet speed. 
%This may be explained by the positive correlation between income and proportion of the bachelor degrees shown in Figs \ref{appendix-fig-corrplot} and \ref{appendix-fig-corrplot_android_CD} for iOS and android devices, respectively. As the 
 Second,   the coefficients of age are negative from both the original and re-sampled data in  City C and  D, suggesting that the regions with younger age tend to have faster internet. 
 %To take into account the interplay between predictors, we analyze the result in Fig.~\ref{appendix-fig-coef-compare} along with the correlation heat map of predictors in Fig.~\ref{appendix-fig-corrplot} for iOS devices from both City C and D. 
Third, we observe in part (a) of Fig.~\ref{appendix-fig-corrplot} that the proportion of Hispanic residents is negatively correlated with the percentage of bachelor's degree or higher, which can partially explain the increased coefficient of Hispanic after re-sampling along with the coefficient of bachelor's degree, particularly for iOS devices. Furthermore,  the internet penetration rate has a positive correlation with income and proportion of individuals with bachelor's degrees or higher, as shown in  Fig.~\ref{appendix-fig-corrplot} and Fig. \ref{appendix-fig-corrplot_android_CD}. Such correlation can potentially explain the  coefficients of the internet penetration rates for iOS devices for City C and City D in the re-sampled data, as positive coefficients of bachelor and income offsets the impact of internet penetration rates. 
%This correlation structure is consistent for Android devices as shown in Fig.~\ref{appendix-fig-boxplot-android-CD} and can be relatable to the fact that the predictor for the internet appears after re-sampling with the decreased coefficients of income and bachelor's degree. 

The results from the impact of demographic profiles in Cities C and  D are generally consistent with Cities A and  B.  Census block groups with higher income and younger population tend to have higher internet speed. Such effects can be offset by other positively correlated variables, such as internet penetration rates. 
%%already discuss before
%The correlation heat map of re-sampled data from Android devices in City A and B is provided in Fig.~\ref{appendix-fig-boxplot-android}. We find substantially positive correlation between income and bachelor's degree, which is consistent with the case from iOS devices as shown in Fig.~\ref{fig-corrplot}. We also note that in City A, the percentage of black residents is negatively correlated with bachelor's degree and the proportion of white population.

\subsection{Temporal progression of  internet speed}\label{appendix-subsec:temporal}

\begin{figure}[t!]
  \centering
  \includegraphics[width=\linewidth]{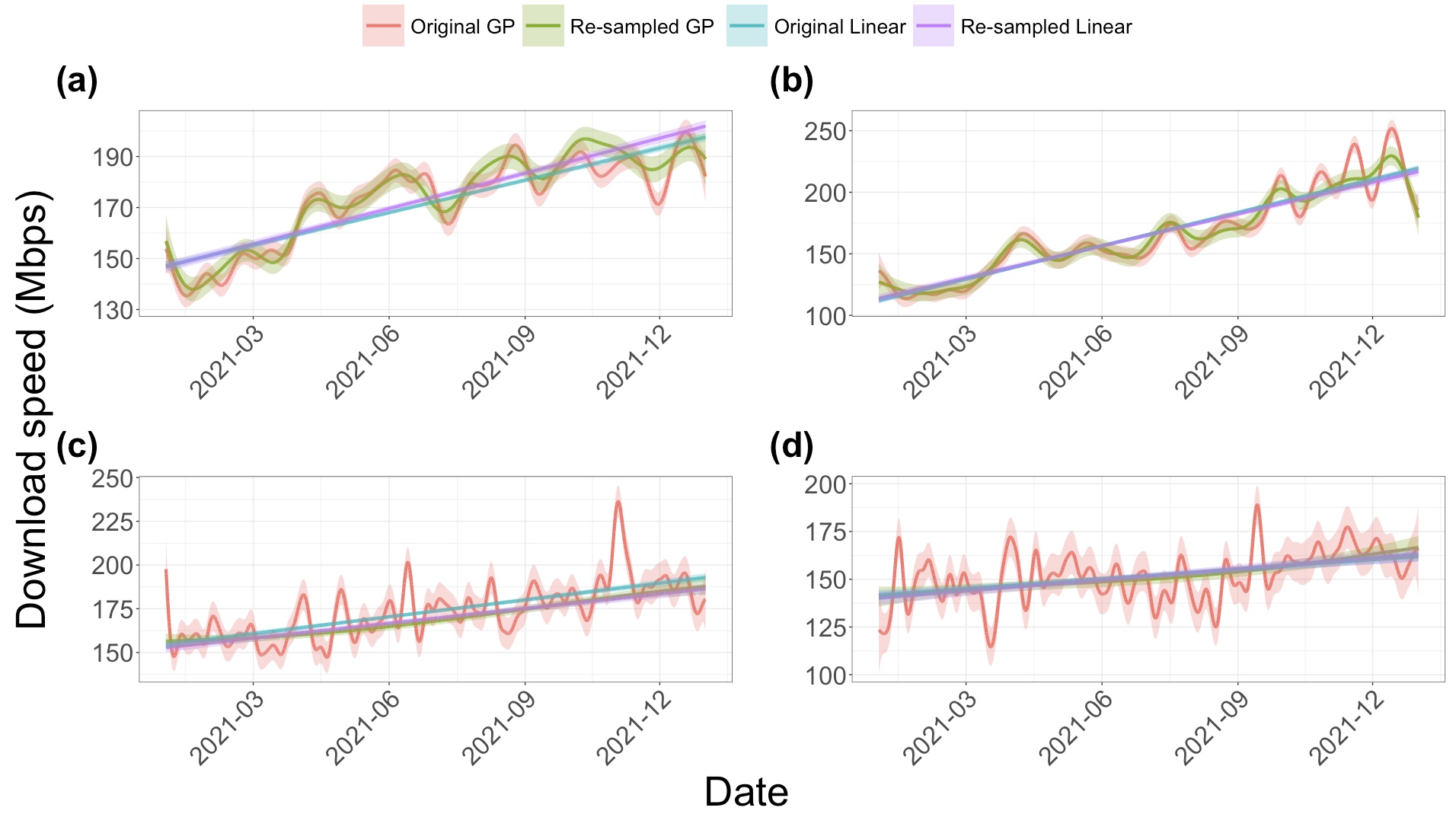}
  \caption{\label{appendix-fig-time-trend}
  Comparison of temporal trend of internet download speed from original versus re-sampled data based on linear and Gaussian process (GP) regression represented by the state space model. (a) iOS devices in City C; (b) Android devices in City C; (c) iOS devices in City D; and (d) Android devices in City D. The solid curves are the predictive mean and the shaded areas are the 95\% percent confidence  interval of the estimation. }
\end{figure}

We analyze the temporal progression of internet speed for Cities C and D based on methods discussed in Section~\ref{sec:temporal}, where in this case the available data spans from 01-01-2021 to 12-31-2021 for both cities. 
 The estimated mean and $95\%$ confidence intervals from both linear models and state space models  are shown in Fig.~\ref{appendix-fig-time-trend}. All models indicate that internet download speed increases over the 1-year period.  %While the models in Fig.~\ref{appendix-fig-time-trend} support the improving trend of internet speed over time, the state space model from original data shows the common surge of internet speed between March 2021 and May 2021. 
 Estimation from original data and re-sampled data seems to be similar for these two cities, indicating that the regional sampling bias across census block groups does not lead to large difference in estimation of temporal trend. 
 %The trends from original data in City C is not significantly different from the re-sampled data for both linear and state space models. We find in City D that space state models from original data shows large volatility over time, suggesting that the original data contains a lot of noise. 
 The data for both cities contains large variability, potentially due to the different internet subscription plans. The availability of fiber internet plan can be one of the main  reasons to drive the increase of internet speed. 
% Section~\ref{sec:conclusion}, internet subscription plans can be a latent factor for this observed noise.

%We first focus on estimating linear trend of temporal improvement of internet speed in Table~\ref{appendix-tab:time-lm}. We found that  increasing trend of internet speed for both iOS and Android devices in City C and D. We observe that re-sampling method tend to alleviate this linear trend except Android devices in City C, and the coefficient decreased significantly for iOS devices in both City C and D.

\begin{table*}[hbt!]
\centering
  \caption{Linear trends of measured internet download speed (Mbps) per day for Cities C and D.}
  \label{appendix-tab:time-lm}
  \begin{tabular}{cccc}
    \toprule
      City & Device type & \thead{Estimates $\hat{\beta}_{1t}$\\(95\% CI)} & \thead{Estimates $\hat{\beta}_{1t}^{r}$\\(95\% CI)} \\
    \midrule
       \multirow{2}{*}{C} 
       &  iOS & 
       \thead{0.1384 \\ (0.1310, 0.1458)} & 
       \thead{0.1509 \\ (0.1403, 0.1615)}\\
    
          & Android & 
          \thead{0.2944 \\ (0.2827, 0.3062)} &  
          \thead{0.2824 \\ (0.2659, 0.2988)}\\
     \midrule
     \multirow{2}{*}{D} 
     &  iOS & 
     \thead{0.1052 \\ (0.0950, 0.1154)} & 
     \thead{0.0919 \\ (0.0786, 0.1053)}\\
    
          & Android & 
          \thead{0.0552 \\ (0.0428, 0.0676)} &  
          \thead{0.0633 \\ (0.0460, 0.0805)}\\
    \bottomrule
  \end{tabular}
\end{table*}

We further compare the estimated linear coefficients  $\hat{\beta}_{1t}$ and $\hat{\beta}_{1t}^{r}$ of the linear trend of internet download speed for City C and City D in Table~\ref{appendix-tab:time-lm}. All estimates are positive for both device types, 
whereas City C may have a faster increasing trend than City D.

\end{document}